\documentclass[10pt,conference,compsocconf,letterpaper]{IEEEtran}

\usepackage[dvips]{graphicx}
\usepackage[usenames,dvips]{color}

\usepackage{epsfig}
\usepackage[caption=false,font=footnotesize]{subfig}
\usepackage{fixltx2e}
\usepackage[cmex10]{amsmath}
\usepackage{amsfonts,amssymb,amsthm}
\usepackage{enumerate}
\usepackage{paralist}
\usepackage[nocompress]{cite}
\usepackage{xspace}
\usepackage{comment}

\usepackage[ruled,lined,linesnumbered]{algorithm2e}

\usepackage{caption}
\usepackage{graphicx}

\usepackage{multirow}

\usepackage{psfrag}

\newtheorem{theorem}{Theorem}[section]
\newtheorem{proposition}[theorem]{Proposition}
\newtheorem{lemma}[theorem]{Lemma}

\newcommand{\ie}{{\em i.e.}\xspace}
\newcommand{\eg}{{\em e.g.}\xspace}

\newcommand{\whp}{{\em w.h.p.}\xspace}

\def\problem{\textsc{MaxLSP}\xspace}
\def\problemu{\textsc{MaxLSP$^{\textrm{U}}$}\xspace}
\def\minproblem{\textsc{MinLSP}\xspace}

\newcommand{\hmm}[1]{}

\begin{document}

%\title{Distributed Algorithms for Maximum Link Scheduling \\in the Physical Interference Model}
\title{Efficient Algorithms for Maximum Link Scheduling \\in Distributed Computing Models with SINR Constraints}

\author{
\IEEEauthorblockN{
Guanhong Pei\IEEEauthorrefmark{1},
Anil Kumar S. Vullikanti\IEEEauthorrefmark{2}
}

\IEEEauthorblockA{\IEEEauthorrefmark{1}Dept.\ of Electrical and
Computer Engineering and Virginia Bioinformatics Institute, Virginia Tech,
Blacksburg, VA 24061}

\IEEEauthorblockA{\IEEEauthorrefmark{2}Dept.\ of Computer
Science and Virginia Bioinformatics Institute, Virginia Tech,
Blacksburg, VA 24061}
}

% make the title area
\maketitle

%=========================================================================
%  Abstract
%=========================================================================

\begin{abstract}
A fundamental problem in wireless networks is the
\emph{maximum link scheduling} (or maximum independent set) problem: given a set $L$ of links, compute the
largest possible subset $L'\subseteq L$ of links that can be scheduled
simultaneously without interference.
This problem is particularly challenging in the
physical interference model based on SINR constraints (referred to as
the SINR model), which has gained
a lot of interest in recent years. Constant factor approximation algorithms
have been developed for this problem, but low complexity distributed
algorithms that give the same approximation guarantee
in the SINR model are not known.
\iffalse
In recent years, there has been a lot of
interest in the physical interference model based on SINR constraints, and
%Goussevskaia et al. (IEEE INFOCOM 2009) and Wan (WASA 2009) have developed
centralized constant factor approximation algorithms are known for this problem.
\fi
Distributed algorithms are especially challenging in this model, because
of its non-locality.

In this paper, we develop a set of fast distributed
algorithms in the SINR model, providing constant approximation
for the maximum link scheduling problem under uniform power assignment.
We find that different aspects of available technology, such as
full/half-duplex communication, and non-adaptive/adaptive power control,
have a significant impact on the performance of the algorithm; these issues have
not been explored in the context of distributed algorithms in the SINR
model before.
Our algorithms' running time is $O(g(L) \log^c m)$,
where $c=1,2,3$ for different problem instances,
and $g(L)$ is the ``link diversity''
determined by the logarithmic scale of a communication link length.
Since $g(L)$ is small and remains in a constant range in most cases,
our algorithms serve as the first set of ``sublinear'' time distributed solution.
The algorithms are randomized and crucially use physical carrier sensing in distributed communication steps.
\end{abstract}

%\IEEEpeerreviewmaketitle

%=========================================================================
%  Introduction
%=========================================================================
\section{Introduction}

%para 1: single shot scheduling problem, different interference models
One of the most basic problems in wireless networks is the
\emph{Maximum Link Scheduling} problem (\problem): given a set $L$ of links, compute the
largest possible subset $L'\subseteq L$ of links that can be scheduled
simultaneously without conflicts; this is
also referred to as the one-shot scheduling \cite{Goussevskaia+:Mobihoc07}
or max independent link set problem \cite{Wan+:WASA09}. One of the main challenges
for this problem is \emph{wireless interference}, which limits the subsets of links
that can transmit simultaneously. A commonly used model is based on ``conflict graphs''
\cite{ramanathan-lloyd:schedmultihop}, broadly referred to as \emph{graph-based interference models};
examples of such models include:
the unit disk graph model, the $k$-hop interference model, and the protocol model.
\problem is challenging in these models --- the decision version of this
problem is NP-Complete, but efficient constant factor approximation
algorithms are known for many interference models \cite{ramanathan-lloyd:schedmultihop},
because of their inherent locality. However, graph based models are known to be
inaccurate and an oversimplification of wireless interference.  In recent
years, a more realistic interference model based on SINR constraints
(henceforth referred to as the SINR model) \cite{Goussevskaia+:Mobihoc07, chafekar:mobihoc07}
has gained a lot of interest: a set of links are feasible simultaneously if the signal
to interference plus noise constraints are satisfied at all receivers (see
Section \ref{sec:model} for the formal definition). This is much harder than graph based
interference models because of the inherently non-local and non-linear nature of the model;
only recently constant factor approximation algorithms have
been developed in this model \cite{Goussevskaia+:INFOCOM09,Halldorsson+:ICALP09,Wan+:WASA09}.

Since link scheduling is a common subroutine in many other problems, distributed
algorithms with low complexity are crucial. A commonly studied model for distributed
computing in wireless networks is the ``Radio Broadcast Network (RBN)'' model, in
which the transmissions on two links conflict if the links interfere (in the
corresponding graph-based model);
variations have been studied of this model, depending on capabilities such as
collision detection. Efficient distributed algorithms are known in the RBN model for
\problem, as well as other fundamental problems such as coloring
and dominating set, e.g., \cite{Moscibroda+:PODC05, Schneider+:PODC08, Afek+:sci11}.
A solution to \problem computed in the RBN model might not be feasible in the
SINR constraints (see, e.g., \cite{Goussevskaia+:Mobihoc07, chafekar:sinrcap08}).
Further, the distributed wireless communication mechanism can be quite different.
In other words, a distributed algorithm in the RBN model
cannot be implemented in general in the SINR-based model.
%, if SINR constraints need to be satisfied, and do not lead to efficient feasible solutions to \problem in the SINR model.
Therefore, we need to rethink
the design of distributed algorithms in the SINR model in a fundamentally new way.

%MIS problem, SINR model
In this paper, we focus on distributed algorithms for \problem in the SINR model, which
is defined in the following manner: at each time step of the algorithm, only those links
for which the SINR constraints are satisfied at the receivers are successful.
The goal of the algorithm is to end up with a feasible solution to \problem, whose
size is maximized.
We have to rethink distributed algorithms in the SINR model for \problem because
of the fundamental differences between the graph-based and the SINR interference models.
As mentioned earlier, even centralized algorithms for this problem are much harder in the SINR model,
than in the disk based interference model; recent work by
\cite{Goussevskaia+:INFOCOM09,Halldorsson+:ICALP09,Wan+:WASA09,Kesselheim:SODA11,Halldorsson+:SODA11} gives
constant factor approximation algorithms for various instances of \problem in the SINR model.
%No distributed algorithm is yet known for any instances of \problem in the SINR model.
The centralized algorithms of
\cite{Goussevskaia+:INFOCOM09, Wan+:WASA09, Halldorsson+:ICALP09}
are based on a greedy ordering
of the links, which requires estimating the ``affectance,'' (which, informally,
is a measure of interference), at each
stage (this is discussed formally later in Section~\ref{sec:model}) --- this
is one of the challenges in distributed solutions to \problem.
We note that efficient time distributed algorithms for scheduling all the links
(i.e., the coloring version) is already known \cite{Halldorsson+:ICALP11}. Adapting them would
immediately yield a distributed $O(\log{m})$-approximation to \problem, but it
is not clear how to obtain a distributed $O(1)$-approximation. Further,
an important aspect of \problem is that the senders and receivers of
all the links should know whether they have been chosen, since this
is an important requirement in many networking applications;
this seems to be
difficult to ensure through random access based approaches.

\subsection{Contributions}
In this paper, we develop fast distributed constant factor
approximation algorithms for \problem, in which all nodes are constrained to
use uniform power levels for data transmission
(we refer to this as \problemu),
improving upon the results implied by \cite{Asgeirsson+:INFOCOM11,Dinitz:INFOCOM10}.
Our algorithms and the proofs build on ideas from
\cite{Wan+:WASA09, Halldorsson+:ICALP09,
Goussevskaia+:INFOCOM09} and \cite{scheideler+:mobihoc08}, and one of the
key technical contributions of our work is the notion of an ``$\omega_1,\omega_2$-ruling'' (discussed below) and
its distributed computation in the SINR model.
Our results raise two new issues in the context of distributed
algorithms in the SINR model --- adaptive power control (i.e., the feature of
using lower than the maximum power level, as needed), and full/half
duplex communication (i.e., whether nodes can transmit and receive
simultaneously). We find these features impact the performance of the
algorithms quite a bit.
We summarize some of the key aspects of the results and main challenges below.
\begin{asparaenum} [(1)]
\item
\emph{Performance and technology tradeoffs}. In the case of ``non-adaptive power control'' (i.e.,
if all nodes are required to use fixed uniform power levels), we design
a distributed algorithm that provably runs in $O(g(L)\log^3{m})$ time
and gives an $O(1)$ approximation to the optimum solution for half duplex communication,
and we improve the running time to $O(g(L)\log^2{m})$ for the case of full duplex communication; here $g(L)$ denotes
the ``link diversity'', which is the logarithm of the ratio of the largest
to the smallest link length (this is defined
formally in Section~\ref{sec:model}).
If nodes are capable of ``adaptive power control'' (i.e., they can use
varying power levels for scheduling, but not data transmission), we improve
the running time of the above algorithm to $O(g(L)\log^2{m})$ time for
half duplex communication,
and $O(g(L)\log{m})$ time for full duplex communication. Note that in the
adaptive power control case, the algorithm uses varying power levels during
its run, but the links which are selected finally use the fixed uniform
power level for data transmission.

\item
\emph{Key distributed subroutine}. One of our key ideas is the parallelization of the link selection,
which would have require sorting all links, processing a larger set
of links simultaneously, and efficient filtering based on spatial and interference constraints in parallel.
Moreover, it turns out that the usual notion of independence
based on spatially separated nodes is inadequate because of
the spatial separation of the sender and the receiver of a link:
it is the senders which makes the distributed decision of transmission and the participation in the independent set,
while the SINR model is receiver-oriented and it is hard for each sender of a candidate link to determine the interference caused by the chosen links at the corresponding receiver.
One of the important steps of our algorithm involves the distributed
construction of a ``ruling'' (a spatially-separated node cover, first introduced in \cite{Cole+:stoc86}) which
relates to the notion of independence and aids the solution to MIS and coloring
problems in graph topologies \cite{Awerbuch+:focs89, Panconesi+:stoc92}.
The extension of the notion of ruling and its computation in the SINR model
is one of the important technical contributions of our paper.
We believe this basic construct would be useful
in other link and topology control problems.
% in SINR model.

\item
\emph{Sensing-based message-less distributed computing}.
We make crucial use of physical carrier sensing, and
in solving \problemu we let the wireless nodes make distributed decisions purely based on the
Received Signal Strength Indication (RSSI) measurement without the need of exchanging or decoding any messages.
Given a threshold $Thres$, a node is able to detect if the total sensed power strength is $\geq Thres$.
As discussed in \cite{scheideler+:mobihoc08}, this
can be done using the RSSI measurement possible through the Clear Channel Assessment
capability in the 802.11 standard.
In this way, the protocol is much simplified such that the wireless nodes only need to control the physical layer to access the medium with a certain power or to sense the channel.
Further, our
algorithm uses constant size messages, and all the steps can be implemented
within the model without additional capabilities or assumptions (as those made in \cite{Asgeirsson+:INFOCOM11}).

\end{asparaenum}

\begin{comment}
We consider two closely related throughput maximization problems: the \textbf{Scheduling} problem and the \textbf{Capacity} or the \textbf{One-shot Scheduling} problem.
We say a schedule is \emph{valid} if the transmissions on the links scheduled are successful.
The \textbf{Scheduling} problem is to seek a minimum length of a valid schedule to accommodate a set of link transmission requests;
the \textbf{Capacity} problem, which resembles the \textbf{MIS} problem in a graph-based network model, is to obtain a valid one-shot schedule of maximum number of links transmission requests.
\end{comment}

\subsection{Key Challenges and Comparisons between Models}

\emph{Comparison between interference models}
It is known that solutions to the link scheduling problems developed under the graph-based models can
be inefficient, if not infeasible, under the SINR model.
For instance, Le et al. \cite{Le+:Mobihoc10} show that the longest-queue-first scheme may result in zero throughput under SINR constraints (unlike that in the graph based model) for the case of dynamic traffic.
As for \problem, it is easy to show that when all the transmitters have uniform transmission/interference ranges,
an optimum solution developed under a graph-based model may turn out to be a solution whose size is a fraction of $O\big((\frac{d_{max}}{d_{tx}})^2\big)$ of that of an optimum under the SINR model, where $d_{max}$ is the length of the longest link and $d_{tx}$ is the uniform transmission range.
This is because that given a set of links under the SINR model, as long as all the senders are separated by $c d_{max}$, where $c$ is some constant, all the links form an independent set.
%whereas in a feasible solution under the RBN model, senders need to be separated by a distance of $\Omega(d_{tx})$.
Since we are dealing with an arbitrary topology, $d_{max}$ may be small, leading to a much conservative solution under the SINR model.

\emph{Comparison between distributed computing models}
In light of the huge amount of research on distributed algorithms in the RBN model
for many problems, including \problem (e.g., \cite{Afek+:sci11, pelegbook}),
it is natural to ask if it might be possible to ``reduce'' the SINR model problem to
the RBN model instead of developing new algorithmic techniques. Though it has not been
rigorously proven, results from recent papers suggest this might not be feasible,
or might only yield larger than constant factor gaps. 
For instance, Chafekar et al.
\cite{chafekar:sinrcap08} discuss an instance where the solution in the ``equivalent''
RBN model could be significantly smaller than that in the SINR model; see also
\cite{moscibroda:ipsn07}. Further, the RBN model does not allow for capabilities to
determine the signal strength and make decisions based on that.

\subsection{Organization}
We discuss
%related work in Section~\ref{sec:related} and
the network model and relevant definitions
in Section~\ref{sec:model}.
We present the high-level distributed algorithm in Algorithm~\ref{alg:one-shot} with a constant approximation ratio in Section~\ref{sec:dist-algr-highlevel}.
We introduce and analyze the distributed algorithm to compute a ruling in Section~\ref{sec:dist-algr-ruling}.
In Section~\ref{sec:dist-impl} we show the detailed implementation for each step of the high-level Algorithm~\ref{alg:one-shot};
%and we show that the implementation is correct and finishes in $O(g(L) \log^3 m)$ time;
we present a second method to implement Algorithm~\ref{alg:one-shot} in Section~\ref{sec:dist-impl-var-p}, improving the running time by a logarithmic factor.
%We conclude the paper in Section~\ref{sec:conclusion}.
%Some of the proofs are omitted due to space constraints; the complete version is available online at \cite{techrep}.

\iffalse
\noindent
\textbf{Related Work}.
We elaborate on related work in Appendix~\ref{append:relatedWork}; below is a brief version.
There has been a lot of research on link scheduling and various related problems, because
of their fundamental nature. The link scheduling problems are well understood in graph based interference
models and efficient approximation algorithms are known for many versions.
Link scheduling in the SINR model is considerably harder than in graph based models.
For \problem, recent progress includes
$O(1)$ approx. ratio under any setting with a fixed length-monotone, sub-linear power assignment by
Halld\'{o}rsson and Mitra \cite{Halldorsson+:SODA11} and
the first $O(1)$-approx. algorithm
by Kesselheim \cite{Kesselheim:SODA11}.
Relevant work also includes
\cite{Goussevskaia+:INFOCOM09, Wan+:WASA09,Halldorsson+:ICALP09,Goussevskaia+:Mobihoc07, chafekar:mobihoc07}.
%Wan+:Infocom11
Most of the these algorithms are centralized
and it is not yet clear how to implement them in a distributed manner efficiently.
The closest results to ours are by \'{A}sgeirsson and Mitra \cite{Asgeirsson+:INFOCOM11} and Dinitz
\cite{Dinitz:INFOCOM10}, using game theoretic approaches. Their running time can be much higher than ours, and they require
additional assumptions (such as acknowledgements without any cost), which
might be difficult to realize in the SINR model.
\fi

\section{Related Work}
\label{sec:related}

There has been a lot of research on link scheduling and various related problems, because
of their fundamental nature. Two broad versions of these problems are: scheduling the
largest possible set of links from a given set (maximum independent set),
and constructing the smallest schedule for all the links (minimum length schedule).
These problems are well understood in graph based interference
models and efficient approximation algorithms are known for many versions;
see, e.g., \cite{
%jain:impact, d2model,
ramanathan-lloyd:schedmultihop}.
Distributed algorithms are also known for node and link scheduling (and many related problems)
in the radio broadcast model \cite{luby:stoc85, Afek+:sci11,
%Afek+:DISC11,
Moscibroda+:PODC05, Schneider+:PODC08, Metivier+:SIROCCO09}. These algorithms are typically
randomized and based on Luby's algorithm \cite{luby:stoc85}, and run in synchronous
polylogarithmic time. There are varying assumptions on the kind of information and
resources needed by individual nodes. For instance, \cite{luby:stoc85} require
node degrees at each step (which might vary, as nodes become inactive).
Moscibroda et al. \cite{Moscibroda+:PODC05} develop algorithms that do not require
the degree information, and run in $O(\log^2 n)$ time.
%Some results, e.g. \cite{Schneider+:PODC08, Metivier+:SIROCCO09} use collision detection capabilities.
In recent work, Afek et al. \cite{Afek+:sci11} develop a distributed algorithm for
the maximal independent set problem, which only requires the an estimate of the
total number of nodes, but not degrees.

Link scheduling in the SINR model is considerably harder than in graph based models.
Several papers developed $O(g(L))$-approximations for \problemu,
e.g., \cite{Goussevskaia+:Mobihoc07, chafekar:mobihoc07}, which have been improved
to constant factor approximations by \cite{Goussevskaia+:INFOCOM09, Wan+:WASA09,
Halldorsson+:ICALP09} for uniform power assignments.
Some of these papers use ``capacity'' \cite{Goussevskaia+:INFOCOM09, Halldorsson+:ICALP09}
to refer to the maximum link scheduling; however, we prefer to avoid the term capacity
in order to avoid confusion with the total throughput in a network, which has
been traditionally referred to as the capacity (e.g., \cite{guptakumar:capacity}).
Recently, Halld\'{o}rsson and Mitra \cite{Halldorsson+:SODA11} extend the $O(1)$
approx. ratio to a wide range of oblivious power assignments for both
uni- and bi-directional links (including uniform, mean and linear power assignments). This has been improved
by Kesselheim \cite{Kesselheim:SODA11}, who developed the first $O(1)$-algorithm for
\problem with power control and an thus an $O(\log m)$-algorithm for the minimum length schedule problem.
%Also, Wan et al. \cite{Wan+:Infocom11} propose algorithms for a set of problems
%including maximum link scheduling and minimum length scheduling, with logarithmic
%approximation guarantees.
Most of the results except those using uniform power
assignments, assume unlimited power values; otherwise the results may degrade by
a factor depending on the ratio of the maximum and minimum transmission power values.

Most of the above algorithms for scheduling in the SINR model are centralized
and it is not clear how to implement them in a distributed manner efficiently.
The closest results to ours are by \'{A}sgeirsson and Mitra \cite{Asgeirsson+:INFOCOM11} and Dinitz
\cite{Dinitz:INFOCOM10}, using game theoretic approaches; the former obtains a constant approx. ratio improving over the latter's $O(d_{max}^{2\alpha})$ approximation
for \problemu. Their running time can be much higher than ours, and they require
additional assumptions (such as acknowledgements without any cost), which
might be difficult to realize in the SINR model.
\iffalse
They are the first attempts to solve \problemu under SINR model in a distributed setting.
However, both their results and approaches rely on an unrealistic assumption for a distributed setting in SINR model, that in each slot a sender can immediately get feedback from its receiver on whether its transmission is successful.
It is not a trivial assumption in SINR setting; their algorithms are not fully distributed unless this can get addressed.
Furthermore, their results state that after $m$-polynomial time, the average approx. ratio drops to some constant, \whp;
the lack of performance guarantee for a single slot makes it unclear of how the goal of \problemu is achieved to compute one instance of max link set distributedly, because nodes need to know whether they are chosen.
%The time-average approx. ratio is a weaker notion for \problemu, because finding a constant-approx. solution with time $T$, implies that keeping the solution for another $T$ will results in constant average approx. ratio.
Also, the convergence time is large ($m$-polynomial).
In contrast, our algorithms are the first fully distributed and efficient solution to \problemu; since $g(L)$ remains small in practice, our algorithms compute a deterministic link set distributedly in polylogarithmic time, resulting in an $O(1)$ approx. ratio \whp
%as commented in \cite{Halldorsson+:ICALP11}, this can be seen more appropriately to roughly determine capacity.
\fi

For the minimum length schedule problem (\minproblem) (where one seeks a shortest schedule to have all the links in $L$ transmit successfully) under a length-monotone sub-linear power assignment,
Fangh\"{a}nel et al. \cite{Fanghanel+:ICALP09} develop a distributed algorithm with an approximate ratio of
$O(g(L))$ times a logarithmic factor.
Recently, Kesselheim and V\"{o}cking \cite{Kesselheim+:DISC10} propose an
$O(\log^2 m)$-approximate algorithm for any fixed length-monotone and sub-linear power assignment.
The approx. ratio of that algorithm has been improved to $O(\log m)$ (matching the best performance of known centralized algorithms) by the analysis of
Halld\'{o}rsson and Mitra \cite{Halldorsson+:ICALP11}, who also prove that if all links uses the same randomized strategy, there exists a lower-bound of $\Omega(\log m)$ on the approx. ratio.
However, it is not clear how to use these results for \minproblem to get a constant
factor approximation for \problem, in which the senders and receivers of all links
know their status.
\iffalse
However, this does not provide much information for \problem.
While a solution to \minproblem implies that there is at least one time slot in the schedule such that
at least $\frac{m}{\log m}$ links can be scheduled together, it is hard to find in which slot how many links successfully transmit due to the distributed nature.
Referring to the construction in the proof of the lower-bound in \cite{Halldorsson+:ICALP11},
it is possible that the solution to \problem is of size $\Theta(m)$, while any solution to \minproblem can only suggest that a solution of size $\frac{m}{\log m}$ to \problem exists, thus leaving a gap of $\log m$.

Our approach also distinguish from these algorithms, in that all our control communication is done at physical layer by sensing the medium without need of decoding.
\fi

%With power control, using mean power assignment,the approximation factor becomes $O( \log{g(L)} \log^3 m)$; the bounds change when using uniform and linear power assignments.

%Another related result \cite{scheideler+:mobihoc08} constructs an $O(1)$-approximate node dominating set in a distributed manner within $O(\log n)$ time.
%and \cite{Lotker+:INFOCOM11} \textcolor{red}{description}.

%=========================================================================
%  Technical Part
%=========================================================================

\section{Preliminaries and Definitions}
\label{sec:model}

\begin{table} [htbp]
\centering
\vspace{-0.05in}
\caption{Notation.}\label{tab:notation}
\fontsize{9.5}{12pt}\selectfont
\begin{tabular} {||c|c||c|c||} \hline
$G$ & network graph & $d(u,v)$ & dist. of $u$ and $v$ \\ \hline
$V$ & set of nodes & $L$ & set of links \\ \hline
$n$ & \#nodes & $g(L)$ & link diversity \\ \hline
$m$ & \#links & $OPT()$ & optimum instance  \\ \hline
$\alpha$ & path-loss exponent & $x(l)$ & sender of link $l$\\ \hline
$\beta$ & SINR threshold & $r(l)$ & receiver of link $l$ \\ \hline
$N$ & background noise & $d(l)$ & length of link $l$ \\ \hline
$A$ & affectance & $SP$ & sensed power \\ \hline
%$B(l,d)$ & ball; center: $x(l)$; radius: $d$ & $B(v,d)$ &  \\ \hline
%$B(v,d)$ & ball; center: $v$; radius: $d$ &  \\ \hline
\end{tabular}
\vspace{-0.1in}
\end{table}
\normalsize

%Table~\ref{tab:notation} in Appendix~\ref{append:notation} summarizes the notation symbols we use frequently.
We let $V$ denote a set of tranceivers (henceforth, referred to as nodes)
in the Euclidean plane. We assume $L$ is a set of links with end-points in $V$,
which form the set of communication requests for the maximum link scheduling
problem at any given time, and $|L| = m$.
Links are directed, and for link $l = \big(x(l), r(l)\big)$, $x(l)$ and $r(l)$
denote the transmitter (or sender) and receiver respectively.
For a link set $L'$, let $X(L')$ denote the set of senders of links in $L'$.
Let $d(u,v)$ denote the
Euclidean distance between nodes $u,v$. For link $l$, let
$d(l)=d(x(l),r(l))$ denote its link length. For links $l, l'$,
let $d(l',l)=d(x(l'),r(l))$.
Let $d_{min}$ and $d_{max}$ denote the smallest and the largest transmission link lengths respectively.
%We normalize the link lengths such that $d_{min} = 1$.
Let $B(v,d)$ denote the ball centered at node $v$ with a radius of $d$.
Each transmitter $x(l)$ uses power $P(l)$ for transmission on link $l$; we assume commonly used path loss models \cite{Goussevskaia+:Mobihoc07,
chafekar:mobihoc07}, in which the transmission on link $l$ is possible only if:
\begin{equation}
\label{eqn:tx-fact}
	\frac{\frac{P(l)}{d^{\alpha}(l)}} {N (1 + \phi)} \geq \beta,
\end{equation}
where $\alpha>2$ is the ``path-loss exponent'', $\beta>1$ is the minimum SINR
required for successful reception, $N$ is the background noise, and
$\phi>0$ is a constant (note that $\alpha, \beta, \phi$ and $N$ are all
constants).

We partition the set of transmission links into non-overlapping link classes.
We define \emph{link diversity}
$g(L)=\lceil \log_2 \frac{d_{max}} {d_{min}}\rceil$.
Partition $L = \{L_i\}, i = 1, 2, \ldots, g(L)$, where
each $L_i = \{l \; | \, 2^{i-1} d_{min} \leq d(l) < 2^i d_{min}\}$ is
the set of links of roughly similar lengths.
Let $d_i = 2^i d_{min}$, such that $d_i$ is an upperbound of link length of $L_i$;
and $\forall i, \forall l \in L_i$, we define $\hat{d}(l) = d_i$.
In a distributed environment, nodes use their shared estimates of minimum and maximum possible link length to replace $d_{min}$ and $d_{max}$, as stated in the previous section.
$g(L)$ in most cases $\leq 6 \log 10$ and remains a constant\footnote{
The minimum link length is constrained by the device dimension, empirically at least 0.1 meter;
the maximum link length depends on the type of the network, and is usually bounded by $10^5$ meters.
For example, the Wi-Fi transmission range is below hundred meters and even long-distance Wi-Fi networks
\cite{Patra+:NSDI07}
have an experimental limit of hundred kilometers; in cellular networks the coverage is at most tens of kilometers; the transmission range in Bluetooth or 60GHz networks is smaller. This implies often $g(L) \leq \log 10^6$.};
further, as discussed earlier, each link can compute which link class it belongs to.
The \emph{reverse link} of a link $l$, denoted by $\overleftarrow{l}$, is the same link with transmission direction inverted.
For a link set $L'$, We use $\overleftarrow{L'}$ to denote the set of reverse links of $L'$.

\noindent
\textbf{Wireless Interference}.
We use physical interference model based on geometric SINR constraints (henceforth
referred to as the SINR model), where
a subset $L' \subseteq L$ of links can make successful transmission simultaneously if and
only if the following condition holds for each link $l \in L'$:
\begin{equation}
\label{eqn:sinr}
	\frac{\frac{P(l)}{d^{\alpha}(l)}} {\sum_{l' \in L' \setminus \{l\}} \frac{P(l')}{d^{\alpha}(l',l)} + N} \geq \beta.
\end{equation}
Such a set $L'$ is said to be \emph{independent} in the context.
%where $\alpha, \beta$ and $N$ are as defined earlier.

\noindent
\textbf{The Maximum Link Scheduling Problem (\problem)}.
Given a set of communication requests (links) $L$, the goal of the \problem problem is
to find a maximum
independent subset of links that can be scheduled simultaneously in the SINR model.
%\problem involves both power control and scheduling;
\problemu is an instance of \problem
where links in a solution use a uniform power level for data transmission;
note that this does not necessarily restrict scheduling to uniform power.
In this paper, we use
$OPT(L)$ to denote an optimum solution to the \problemu, and thus
$|OPT(L)|$ is the cardinality of the largest such independent set. As discussed earlier,
computing $OPT(L)$ is NP-hard, and we focus on approximation algorithms. We say an
algorithm gives a $C$-approximation factor if it constructs an independent link set
$L'\subseteq L$ with $|L'|\geq |OPT(L)| / C$.

\noindent
\textbf{Distributed Computing Model in the SINR-based Model}.
Traditionally, distributed algorithms for wireless networks have been studied in the
radio broadcast model \cite{d2model, Moscibroda+:PODC05, Schneider+:PODC08} and its
variants. The SINR based computing model is relatively recent, and has not been
studied that extensively. Therefore, we summarize the main aspects and assumptions
underlying this model:
%We assume $L$ is the set of communication requests.
\begin{inparaenum} [(1)]
\item
The network is synchronized and for simplicity we assume all slots have the
same length.
\item All nodes have a common estimate of $m$, the number of links, within
a polynomial factor;
\item For each link $l\in L$, $x(l)$ and $r(l)$ have an estimate of $d(l)$, but they
do not need to know the coordinates or the direction in which the link is oriented;
\item All nodes share a common estimate of $d_{min}$ and $d_{max}$, the minimum and maximum possible link lengths;
\iffalse
\footnote{
The fourth assumption is not a strong one:
\begin{inparaenum} [(1)]
\item the min. length can be set to the device dimension, empirically at least 0.1 meter;
\item the max. link length depends on the type of the network, and is usually bounded by $10^5$ meters.
\end{inparaenum}
For example, the Wi-Fi transmission range is below hundred meters and even long-distance Wi-Fi networks
\cite{Patra+:NSDI07}
have an experimental limit of hundred kilometers; in cellular networks the coverage is at most tens of kilometers; the transmission range in Bluetooth or 60GHz networks is smaller. This implies often $g(L) \leq \log 10^6$.}.
\fi
\item
We assume nodes have physical carrier sensing capability and can detect if the
sensed signal exceeds a threshold. As discussed in \cite{scheideler+:mobihoc08}, this
can be done using the RSSI measurement possible through the Clear Channel Assessment
capability in the 802.11 standard.
Given a threshold $Thres$, we assume that a node is
able to detect if the sensed power strength is $\geq Thres$.
%Given a threshold $T$, we assume that a node $v$ is able to detect if $N+\sum_{l\in L'} P(l)/d(x(l),v))^{\alpha} \geq T$.
\iffalse
Since we do not use any content of messages,
we do not require any specific packet processing (or decoding) capabilities for the control phase,
in which the distributed coordination is done.
For data transmission, we refer to Inequality~\eqref{eqn:sinr} as the condition of stable or successful transmission, which is the goal of our algorithms.
\fi
\end{inparaenum}

\noindent
\textbf{Sensed Power-strength and Affectance}.
For ease of analysis based on links, we define \emph{affectance}\footnote{Sometimes it shares the same definition with the term \emph{relative interference}, \eg, in \cite{Wan+:WASA09}; however, ``relative interference'' may refer to other forms, \eg, in \cite{Halldorsson+:ICALP09}.} as that in \cite{Halldorsson+:SODA11, Halldorsson+:ICALP09}: the affectance, caused by the sender of link $l'$ to the receiver of link $l$, is
$\displaystyle
A(l',l) = \frac{\beta}{1 - \frac{d^{\alpha}(l)}{P / (\beta N)}} \frac{d^{\alpha}(l)}{d^{\alpha}(l',l)}.
$
Likewise, we have affectance from a set $L'$ of links, as $A(L',l) = \sum_{l' \in L'} A(l',l)$.
%This is also widely referred to as \emph{affectance}.
It can be verified that Inequality~\eqref{eqn:sinr} is equivalent to $A(L' \setminus \{l\}, l) \leq 1$, signifying the success of data transmission on $l$.

To simplify the analysis based on nodes,
we define \emph{sensed power-strength} $SP(w, v)$, as the signal power that node $v$ receives when only $w$ is transmitting (which includes background noise); that is,
$SP(w, v) = P / d^{\alpha}(w,v) + N$.
Likewise, we have SP from a node set $W$: $SP(W,v) = \sum_{w \in W} P / d^{\alpha}(w,v) + N$.
Let $Thres(d) = P / d^{\alpha} + N$ be a function of distance $d$, such that for a node $v$, if any other node is transmitting in a range of $d$, its sensed power will exceed $Thres(d)$.

\begin{comment}
\begin{proposition}
\label{prop:RI-SINR}
$\forall L' \in L, \forall l, \frac{P/d^{\alpha}(l)}{\sum_{l' \in L' \setminus \{l\}} \frac{P}{d^{\alpha}(l',l)} + N} \geq \beta \Leftrightarrow RI(L' \setminus \{l\}, l) \leq 1$.
\end{proposition}

Proposition~\ref{prop:RI-SINR} implies that, a set $L'$ of links can transmit successfully if and only if $\forall l \in L'$:
$
RI(L' \setminus \{l\}, \,l) \leq 1.
$
\end{comment}

\noindent
\textbf{Node Capabilities for Distributed Scheduling}
\begin{inparaenum} [(1)]
\item \emph{Half/full Duplex Communication}:
Wireless radios are generally considered \emph{half duplex}, \ie, with a single
radio they can either transmit or receive/sense but not both at the same time.
\emph{Full duplex} radios, which are becoming reality, %\cite{Jain+:mobicom11,Khojastepour+:HotNets11}.
enable wireless radios to perform transmission and reception/sensing simultaneously.
%through some self-interference canceling design.
\item \emph{Non-adaptive/adaptive Power}:
Although links in a solution to \problemu use a uniform power level for data transmission,
they are usually capable of using adaptive power which vary across different power levels that may be \emph{used for scheduling}. The capabilities can play a vital role in distributed computation.
\end{inparaenum}
%this capability may also play a vital role in the distributed computation.

%it is natural to use the same uniform power level for scheduling purposes that compute solutions to \problem.
%However, wireless nodes are usually capable of using \emph{adaptive power}, which can vary across different power levels; this capability may also play a vital role in the distributed computation.

\section{Distributed Algorithm: Overview}
\label{sec:dist-algr-highlevel}
In this section, we present the distributed algorithm
for \problemu. Because the algorithm is quite complicated, we briefly
summarize the sequential algorithm of \cite{Wan+:WASA09, Goussevskaia+:INFOCOM09,
Halldorsson+:ICALP09} below, and then give a high-level
description of the distributed algorithm and its analysis, without the
implementation details of the individual steps in the SINR model.
Section~\ref{sec:dist-algr-ruling} describes the algorithm for computing a ruling in the full and half duplex models.
The complete distributed implementation and other details are discussed
in Sections~\ref{sec:dist-impl} and~\ref{sec:dist-impl-var-p}, for
the non-adaptive and adaptive power control settings, respectively.

\subsection{The Centralized Algorithm}
\label{sec:appendix-seq}

We discuss the centralized algorithm adapted from
\cite{Wan+:WASA09, Goussevskaia+:INFOCOM09, Halldorsson+:ICALP09} for \problemu,
which forms the basis for our distributed algorithm. The algorithm processes
links in non-decreasing order of length. Let
$L$ be the initial set of links, and $S$ the set of links already chosen
(which is empty initially). Each iteration involves the following steps:
\begin{compactenum} [(1)]
\item picking the shortest link $l$ in $L \setminus S$ and removing $l$ from $L$,
\item removing from $L$ all the links in $\{l' \in L \setminus S : A(S,l') \geq c_0\}$ where $c_0<1$ is a constant, \ie, all the links in $L \setminus S$ that suffer from high interference caused by all chosen links in $S$, and
\item removing from $L$ all the links in $\{l' \in L \setminus S : d(l',l) = c_1 d(l)\}$ where $c_1$ is a constant, \ie, all the nearby links of $l$ in $L \setminus S$.
\end{compactenum}

The results of \cite{Goussevskaia+:INFOCOM09, Wan+:WASA09, Halldorsson+:ICALP09}
show that: $S$ is feasible (i.e., the SINR constraints are satisfied at
every link), and $|S|$ is within a constant factor of the optimum.
Consider a link $l$ that is added to $S$ in iteration $i$.
The proof of feasibility of set $S$ involves showing that
for this link $l$, the affectance due to the links added to $S$ after
iteration $i$ is at most $1-c_0$, so that simultaneous transmission
by all the links in $S$ does not cause high interference for $l$.
The approximation factor involves the following two ideas:
\begin{inparaenum}[(1)]
\item for any link $l\in S$, there can be at most $O(1)$ links in $OPT(L)$ which
are within distance $c_0 d(l)$, and
\item in the set of links removed in step 2 due to the affectance from $S$,
there can be at most $O(1)$ links in $OPT(L)$.
\end{inparaenum}
We note that the second and third steps are reversed in \cite{Wan+:WASA09}, while
\cite{Halldorsson+:ICALP09} does not use the third step. However, we find
it necessary for our distributed algorithm, which uses the natural approach
of considering all the links in a given length class simultaneously
(instead of sequentially).
Our analysis builds on these ideas, and property (1) holds for our case
without any changes. However, property (2) is more challenging to analyze, since
many links are added in parallel. Another complication is that the
distributed implementation has to be done from the senders' perspective,
so that the above steps become more involved.

\subsection{Additional Definitions}
\label{sec:preliminaries}

\begin{figure}[htbp]%
%\centering%
%\centerline{\subfloat[Example of ruling.]{\includegraphics[height=2.6in, width=3in]{figs/ruling-makefig.eps}%
\centering%
\includegraphics[height=2in, width=2.22in]{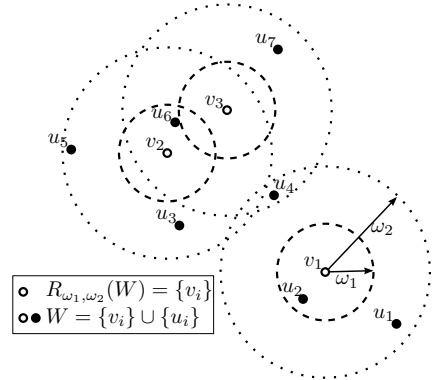}%
\caption{\small
Example of an $(\omega_1,\omega_2)$-ruling:
  $W = \{v_i\} \cup \{u_i\}$ is the set of all dots (open and dark), while
  $R_{\omega_1, \omega_2}(W) = \{v_i\}$ which is the set of all the open dots denotes a $(\omega_1, \omega_2)$-ruling of $W$.
  Note that all the nodes in $W$ are $\omega_2$-covered by $R_{\omega_1, \omega_2}(W)$,
  while all the nodes in $R_{\omega_1, \omega_2}(W)$ are $\omega_1$ away from each other.
\label{fig:example}%
}%
\end{figure}%

\normalsize
\noindent
\textbf{Cover and Ruling}.
Let $W,W'$ denote two node sets.
We say a node $u$ is \emph{$\omega$-covered}~by~$W'$, if and only if $\exists u' \in W', d(u,u') \leq \omega$;
based on that, we say $W$ is \emph{$\omega$-covered} by $W'$, or equivalently $W'$ $\omega$-\emph{covers} $W$, if and only if every node in $W$ is $\omega$-covered by $W'$.
An $(\omega_1,\omega_2)$-\emph{ruling} (where $\omega_1 < \omega_2$) of $W$, introduced in \cite{Cole+:stoc86}, is a node set denoted by $R_{\omega_1, \omega_2}(W)$, such that
\begin{compactenum} [(1)]
\item
$R_{\omega_1, \omega_2}(W) \subseteq W$;
\item
all the nodes in $R_{\omega_1, \omega_2}(W)$ are at least $\omega_1$-separated; that is, $\forall u,u' \in R_{\omega_1, \omega_2}(W)$, $d(u, u') \geq \omega_1$; and
\item
$W$ is $\omega_2$-covered by $R_{\omega_1, \omega_2}(W)$.
\end{compactenum}
Here, we have adopted a generalized definition by considering Euclidean distance rather than graph distance.
The concept of ruling has a vital role in our algorithm: it is used for choosing a set of spatially separated links and removing the nearby links of the chosen links.
%That is, $Z$ consists of two definite components --- all the rest of the nodes in $W$ apart from the ruling, all the nodes in $W'$ that are $\omega_1$-covered by the ruling --- and one indefinite component, \ie, an arbitrary subset of the nodes in $W'$ that are $\omega_2$-covered but not $\omega_1$-covered by the ruling. An important property is that $Z$ is always $\omega_2$-covered. Such a definition is due to the randomized nature of our distributed algorithms.
Figure~\ref{fig:example}
%in the Appendix~\ref{append:notation}
gives an example to illustrate these notions.

\begin{comment}
We extend the above node-based definitions to links, by referring to the sender nodes of links.
For example, let $Y,Y'$ be link sets; we say $Y$ is \emph{$\omega$-covered} by link set $Y'$, if and only if $X(Y)$ is \emph{$\omega$-covered} by $X(Y')$; we say $Y$ is an $(\omega_1,\omega_2)$-\emph{ruling} of $Y'$, if and only $X(Y)$ is an $(\omega_1,\omega_2)$-\emph{ruling} of $X(Y')$.
Here, recall that $X(L')$ denotes the set of senders of links in $L'$.
In the same way, we define the link-version of \emph{extended ruling}.
\end{comment}

\subsection{High-level Description of the Distributed Algorithm}
\label{sec:description}
\begin{algorithm}[t!]
\SetNoFillComment
\SetKwComment{Comment}{/* }{ */}
\SetKwInOut{InputP}{input}
\SetKwInOut{OutputP}{output}

\InputP{Set $L$ of links}
\OutputP{One-shot Schedule $S$}
\vspace*{0.05in}

$J \gets L$\;

%\Comment{For each link class}
\ForEach(\tcc*[h]{for each link class in $L$}){$i = 1,2,\ldots,g(L)$}{
    %\Comment{Initialization for iteration $i$ (Lines~\ref{algline:step0-start} and~\ref{algline:step0-end})}
	$J_i \gets L_i \cap J,\, J^{>}_i \gets \cup_{j > i} L_j \cap J$,
	$\omega_1 \gets \gamma_1 d_{i},
	\omega_2  \gets \gamma_2 d_{i}$\; \label{algline:step0}
	%\vspace*{0.05in}
    \fontsize{8.5}{10.5pt}\selectfont
	\Comment{\textbf{1st step (Lines~\ref{algline:ri-start}-\ref{algline:ri-end})}: check affectance constraints, s.t. links in $S_i$ are not subject to high interference by links of smaller lengths chosen in previous phases. Note: we use reverse $l$ to check affectance at $x(l)$}
    \normalsize
	\If{$i > 1$} {\label{algline:ri-start}
		$J^a_i \gets \big\{l \in J_i: \, A(\cup_{j < i} S_j, \overleftarrow{l}) \leq \psi (1 - (\frac{\phi}{\beta (1+\phi)})^{1/\alpha})^{\alpha}\big\}$, $\overline{J^a_i} \gets J_i \setminus J^a_i$\;
		$J^{b}_i \gets \big\{l \in J^{>}_i: \, A(\cup_{j < i} S_j, \overleftarrow{l}) \leq \psi (1 - (\frac{\phi}{\beta (1+\phi)})^{1/\alpha})^{\alpha}\big\}$, $\overline{J^b_i} \gets J^{>}_i \setminus J^b_i$\;
	} \label{algline:ri-end}
	%\vspace*{0.05in}
    \fontsize{8.5}{10.5pt}\selectfont
	\Comment{\textbf{2nd step (Lines~\ref{algline:ruling} \&~\ref{algline:byproduct})}: check spatial constraints, to obtain an $(\omega_1,\omega_2)$-ruling $X(J^r_i)$ of $X(J^a_i)$, s.t. the selected links ($J^r_i$) are spatially separated and nearby links ($J^z_i$) of similar or larger lengths are excluded}
    \normalsize
	construct link set $J^r_i$, s.t. $X(J^r_i)$ is an $(\omega_1,\omega_2)$-ruling of $X(J^a_i)$\; \label{algline:ruling}
%\Comment*[l]{\fontsize{8.5}{10.5pt}\selectfont to be selected\normalsize} \label{algline:ruling}
    construct link set $J^z_i$, s.t.
    (1) $J^z_i \cap J^a_i = J^a_i \setminus J^r_i$ and (2)~$\{l \in J^b_i : x(l) \text{ is } \omega_1\text{-covered by } X(J^r_i)\} \subseteq J^z_i \cap J^b_i \subseteq \{l \in J^b_i : x(l) \text{ is } \omega_2\text{-covered by } X(J^r_i)\}$\; \label{algline:byproduct}
    %\Comment*[l]{\fontsize{8.5}{10.5pt}\selectfont nearby links of similar and larger lengths to be removed\normalsize}
	%\vspace*{0.05in}	
    \fontsize{8.5}{10.5pt}\selectfont
	\Comment{\textbf{3rd step (Lines~\ref{algline:step3-1} \&~\ref{algline:step3-2})}: select and discard links}
    \normalsize
	$S_i \gets J^r_i$, $J \gets J \setminus J^r_i$\Comment*[l]{\fontsize{8.5}{10.5pt}\selectfont select links\normalsize} \label{algline:step3-1}
	$J \gets J \setminus (\overline{J^a_i} \cup \overline{J^b_i} \cup J^z_i)$\Comment*[l]{\fontsize{8.5}{10.5pt}\selectfont discard links\normalsize} \label{algline:step3-2}
}
%\vspace*{0.05in}

$S \gets \cup_i S_i$, Return $S$\;

\caption{Distributed Maximum Link Scheduling\label{alg:one-shot}}
\end{algorithm}

We have provided a detailed discussion of the centralized algorithms and basic ideas in Section~\ref{sec:appendix-seq};
now we discuss the distributed algorithm at a high level
(Algorithm~\ref{alg:one-shot}), and prove the main properties.
We use the following constants in the algorithm
$\gamma_1 = \left(\frac{36 \beta}{1-\psi}\frac{\alpha-1}{(\alpha-2)}\frac{1+\phi}{\phi}\right)^{1/\alpha} + 2$ and
$\gamma_2$ as an arbitrary constant $>\gamma_1$, where $\alpha,\beta,\phi$ are constants in described in Inequality~\eqref{eqn:tx-fact} and $\psi$ is a constant that can take any value from $(0,1)$.
The algorithm sweeps through the link classes in $g(L)$ phases.
In the $i$th phase, where $i \in [1,g(L)]$, it selects a subset of links from $L_i$ to include in $S$, and removes a subset of links from $\cup_{j>i} L_{j}$ to speed up later phases;
the comments in Algorithm~\ref{alg:one-shot} explain each step.
\begin{compactenum}[(1)]
\item Step 1 in line \ref{algline:ri-start}-\ref{algline:ri-end} of the algorithm eliminate links which do not satisfy affectance
constraints --- its implementation is formally described in
Sections~\ref{sec:dist-impl} and~\ref{sec:dist-impl-var-p}, and
crucially relies on the measurement of the received power at each sender.
\item Step 2 in line \ref{algline:ruling} of the algorithm constructs the distributed ruling,
which is discussed in Section~\ref{sec:dist-algr-ruling} for the
non-adaptive power case, and extended to the adaptive case in
Section~\ref{sec:dist-impl-var-p}.
\end{compactenum}

\begin{lemma}[Correctness]
\label{lemma:correctness}
The algorithm results in an independent set $S$.
%$RI(S \setminus \{l\}, l) \leq 1$ for each $l \in S$.
\end{lemma}

Lemma~\ref{lemma:correctness} shows the correctness and is proved in Appendix~\ref{append:correctness}.
Algorithm~\ref{alg:one-shot} results in a constant approximation ratio shown in Theorem~\ref{theorem:constantS}; its proof is in Appendix~\ref{append:constantS}, where this theorem is backed by Lemma~\ref{lemma:constant-step1} and~\ref{lemma:constant-step2} in Appendix~\ref{append:constantS}.
The two lemmas, independent of Algorithm~\ref{alg:one-shot}, stands on its own to provide insights of how an optimum solution is shaped under the SINR model; they can be proved by using a combination of techniques found in \cite{Wan+:WASA09,Halldorsson+:ICALP09}. We provide their proofs in Appendix~\ref{append:constantS} to make the paper complete and help the reading.
%\textcolor{red}{The basic idea is to show that in each iteration, the number of links we discard is small.}

\begin{theorem}[Approximation Ratio]
\label{theorem:constantS}
$\forall \gamma_1, \gamma_2, \psi > 0$, $|OPT(L)| \leq C_3(\gamma_2, \psi) |S|$ if $\gamma_2 > \gamma_1 > 1$, where $C_3(\gamma_2, \psi) = C_1(\gamma_2) + C_2(\psi (1 - (\frac{\phi}{\beta (1+\phi)})^{1/\alpha})^{\alpha})$,
where $C_1(x) = \frac{(2 x+1)^{\alpha}} {\beta}$ and
$C_2(x) = (\frac{2(\beta b)^{1/\alpha}}{(\beta b)^{1/\alpha} - 1})^{\alpha} / x+1$.
are functions with constant output values for constant input arguments.
\end{theorem}

%When $\gamma_1, \gamma_2, \psi$ are constants, Theorem~\ref{theorem:constantS} implies that $|S| = \Omega(|OPT(L)|)$.

\section{Distributed Algorithm: $(\omega_1,\omega_2)$-Ruling}
\label{sec:dist-algr-ruling}

In this section, we present Algorithm~\ref{alg:ex-ruling}, the distributed algorithm to compute an $(\omega_1,\omega_2)$-ruling, for full duplex communication under the physical interference model; in the end of the section, we extend it to the half duplex setting (where a node can perform transmission and reception/sensing at the same time) with added running time. While this algorithm can be interesting by itself, it serves as a significant building block for our distributed implementation.
For the algorithm to function properly, we require the input parameter $\omega_2 \geq (36 \frac{\alpha-1}{\alpha-2})^{\alpha-2} \omega_1$.
Recall that $B(v,d)$ denotes the ball centered at $v$ with a radius of $d$.
Let $n$ be the total number of nodes.
The last input parameter $b_{max}$ denote the estimate of the maximum number of nodes in the ball $B(v, \omega_1)$ of any node $v \in W_1$; in the worst case, $b_{max} \leq n$.

\begin{algorithm}[ht!]
\SetNoFillComment
\SetKwComment{Comment}{/* }{ */}
\SetKwInOut{InputP}{input}
\SetKwInOut{OutputP}{output}

\InputP{$\omega_1,\omega_2,W_1,W_2,b_{max}$}
%\Comment{$\omega_1 = \gamma_1 d_i, \omega_2 = \gamma_2 d_i, X = J_i, Y = J_{>i}$ in Algorithm~\ref{alg:one-shot}}
\OutputP{$(\hat{R}, \hat{Z})$: an $(\omega_1,\omega_2)$-ruling of $W_1$}
%\Comment{$R = R_i, Z = Z_i$ in Algorithm~\ref{alg:one-shot}}
\vspace*{0.05in}

%\Comment{$\hat{R}(v),\hat{Z}(v)$: indicators of \\ whether node $u$ belongs to $\hat{R}$ or $\hat{Z}$}
%$\hat{R}(v) \gets \text{false}, \hat{Z}(v) \gets \text{false}$\;
%\vspace*{0.05in}

\fontsize{8.5}{10.5pt}\selectfont
\Comment{Each $v \in W_1 \cup W_2$ does the following}
\normalsize
%\vspace*{0.05in}

\For{$i_{out} = 0$ to $\log b_{max} + 1$} {\label{algline:loop1}
    \For{$i_{in} = 1$ to $C_4 \log n$} {\label{algline:loop2}
    	\If{$v$ is active} {
            \fontsize{8.5}{10.5pt}\selectfont
    		\Comment{Coordination Step (Lines~\ref{algline:coor-start}-\ref{algline:coor-end}): 1 slot}\normalsize
    		$U(v) \gets 0$\; \label{algline:coor-start}
	    	\lIf{$v \in W_1$}{$U(v)$ flips to 1 w/ prob. $\frac{2^{i_{out}-2}}{b_{max}}$}\; \label{algline:inU}
	    	\vspace*{0.05in}
            \If{$U(v) = 1$} { \label{algline:loop3}
            	$v$ transmits and senses, \,$I(v) \gets$ the power $v$ receives in this slot\;
            	\lIf{$I(v) > Thres(\omega_1)$} {
            		$U(v) \gets 0$\;
            	}
            }  \label{algline:coor-end}
	    	%\vspace*{0.05in}
            \fontsize{8.5}{10.5pt}\selectfont
	    	\Comment{Decision Step (Lines~\ref{algline:decision-start}-\ref{algline:decision-end}): 1 slot}\normalsize
	    	\lIf{$U(v) = 1$} { \label{algline:decision-start}
	    		$v$ transmits, \,$v$ joins $\hat{R}$\Comment*[l]{\fontsize{8.5}{10.5pt}\selectfont inactive\normalsize}
	    	}
	        \Else{
	    		$v$ senses, $I(v) \gets$ the power $v$ receives in this slot\;
	    		\lIf{$I(v) > Thres(\omega_1)$} {
	    			$v$ joins $\hat{Z}$\Comment*[l]{\fontsize{8.5}{10.5pt}\selectfont inactive\normalsize}
	    		}
	    	} \label{algline:decision-end}
	    }
    }
}
\Return $(\hat{R}, \hat{Z})$\;
\caption{$\operatorname{ConstructR}(\omega_1,\omega_2,W_1,W_2,b_{max})$: Distributed algorithm for computing an $(\omega_1,\omega_2)$-ruling with full duplex radios.\label{alg:ex-ruling}}
\end{algorithm}

In this algorithm, we call an iteration of the outer loop (Line~\ref{algline:loop1}) a \emph{phase}; we call an iteration of the inner loop (Line~\ref{algline:loop2}) a \emph{round}, consisting of the \textit{coordination step} (Lines~\ref{algline:coor-start} through~\ref{algline:coor-end}) and the \textit{decision step} (Lines~\ref{algline:decision-start} through~\ref{algline:decision-end}).
A node $v$ is said to be \emph{active} if $v$ has not joined either $\hat{R}$ or $\hat{Z}$; otherwise, $v$ becomes \emph{inactive}.

In each round, the coordination step provides a probabilistic mechanism for active nodes to compete to get in the ruling (at Line~\ref{algline:inU}).
\begin{comment}
Lines~\ref{algline:loop3} through~\ref{algline:coor-end} constitute a module to resolve the issue of sensing and transmitting at the same time, such that two nearby nodes do not both enter the ruling (\ie, Lemma~\ref{lemma:spaced}).
To help towards a better understanding, suppose there existed a miracle that enables full duplex for concurrent sensing and transmitting, or a device that can quickly switch from transmission to sensing mode to outspeed the signal transmission, we can replace this $C_5 \log n$-time module with a one-slot miracle, such that Lemma~\ref{lemma:spaced} is automatically true.
\end{comment}
Lines~\ref{algline:loop3} through~\ref{algline:coor-end} constitute a module to resolve the issue of sensing and transmitting at the same time, such that two nearby nodes do not both enter the ruling (\ie, Lemma~\ref{lemma:spaced-full}).
Next, during the decision step, a subset of active nodes decide to join $\hat{R}$ or $\hat{Z}$.

In each phase, there are $C_4 \log n$ rounds, such that we can ensure a fraction of the node population have either joined $\hat{R}$ or $\hat{Z}$, and we expect the maximum number of active nodes in the nearby region of any active node to decrease by a half (proved in Lemma~\ref{lemma:degree} in Appendix~\ref{append:runtime}).
After each phase, the probability for each active node to access the channel and compete doubles (at Line~\ref{algline:inU}). After the total of $\log b_{max} + 2$ phases, we have Lemmas~\ref{lemma:runtime},~\ref{lemma:covered1},~\ref{lemma:covered2} that lead to Theorem~\ref{theorem:alg-ex-ruling}.

%\subsection{Correctness}
%of Algorithms~\ref{alg:ex-ruling}}
\label{sec:correctness-dist-algr}
%By Theorem~\ref{theorem:alg-ex-ruling} and Lemma~\ref{lemma:alg-step2}, we show that our distributed algorithms functions as claimed.

\begin{theorem}[Correctness]
\label{theorem:alg-ex-ruling}
%With probability of at least $1 - \frac{1}{n^{C_{11}}}$,
Algorithm~\ref{alg:ex-ruling} terminates in $O(\log n \log b_{max})$ time.
By the end of the algorithm:
\begin{inparaenum} [(1)]
\item $\hat{R}$ forms an $(\omega_1, \omega_2)$-ruling of $W_1$, and
\item $\hat{Z} \cap W_1 = W_1 \setminus \hat{R}$ and $\{v \in W_2 : v \text{ is } \omega_1\text{-covered by } \hat{R}\} \subseteq \hat{Z} \cap W_2 \subseteq \{v \in W_2 : v \text{ is } \omega_2\text{-covered by } \hat{R}\}$,~\whp
\end{inparaenum}
\end{theorem}

Theorem~\ref{theorem:alg-ex-ruling} follows directly from the lemmas below.
Lemmas~\ref{lemma:runtime},~\ref{lemma:spaced-full} and~\ref{lemma:covered2} prove that $\hat{R}$ is an $(\omega_1, \omega_2)$-ruling of $W_1$, \whp
Lemmas~\ref{lemma:runtime},~\ref{lemma:covered1},~\ref{lemma:covered2} together shows that $\hat{Z}$ complements $\hat{R}$ in $W_1$ and partially in $W_2$ with desired properties, \whp
To help the reading flow and due to the page limit, we defer much of the technical content --- the proof of Lemma~\ref{lemma:runtime} (which involves Lemmas~\ref{lemma:lucky},~\ref{lemma:success} and~\ref{lemma:degree}) and the proof of Lemma~\ref{lemma:covered2}
%, and the proof of Lemma~\ref{lemma:spaced-full}
--- to Appendix~\ref{append:ruling}.
%in the online technical report~\cite{pei+:icdcs13}.

\begin{comment}
Theorem~\ref{theorem:alg-ex-ruling} follows directly the lemmas below.
Lemmas~\ref{lemma:runtime} and~\ref{lemma:spaced} support Theorem~\ref{theorem:alg-ex-ruling}-(1) for $\hat{R}$, by showing that $W_1 \setminus \hat{R}$ is $\omega_2$-covered by $\hat{R} \subseteq W_1$, all nodes in $\hat{R}$ are $\omega_1$-separated, \whp
Lemmas~\ref{lemma:runtime},~\ref{lemma:covered1},~\ref{lemma:covered2} together yield Theorem~\ref{theorem:alg-ex-ruling}-(2) for $\hat{Z}$, by showing that $\hat{Z}$ includes all nodes in $W_1 \setminus \hat{R}$ and all $\omega_1$-covered nodes in $W_2$, plus that $\hat{Z}$ is $\omega_2$-covered by $\hat{R}$, \whp
To help the reading flow, we defer much of the technical content --- the proof of Lemma~\ref{lemma:runtime} (which involves Lemmas~\ref{lemma:lucky},~\ref{lemma:success} and~\ref{lemma:degree}), and the proof of Lemmas~\ref{lemma:spaced} --- to Appendix~\ref{append:ruling}.
\end{comment}

\begin{lemma}[Completion]
\label{lemma:runtime}
By the end of the algorithm, all nodes in $W_1$ have joined either $\hat{R}$ or $\hat{Z}$, \ie, all nodes in $W_1$ become inactive, \whp
%with probability of at least $1 - \frac{1}{n^{C_7-1}}$, where $C_7$ is a constant.
\end{lemma}

Lemma~\ref{lemma:runtime} implies that $\hat{Z} \cap W_1 = W_1 \setminus \hat{R}$.
We say a node $v \in \hat{R}$ is ``\emph{good},'' if and only if $d(v,v') \geq \omega_1, \forall v' \in \hat{R}$ and $v' \neq v$. In Algorithm~\ref{alg:ex-ruling}, When a node enters $\hat{R}$, it makes sure that there are no other ones entering $\hat{R}$ within a range of $\omega_1$, and it deactivate all the active nodes in the same range. Therefore, we have the following Lemmas~\ref{lemma:spaced-full} and~\ref{lemma:covered1}.
\begin{lemma}[Quality of $\hat{R}$]
\label{lemma:spaced-full}
All nodes in $\hat{R}$ are good, with probability of 1.
\end{lemma}

\begin{lemma}[Quality of $\hat{Z}$: Part 1]
\label{lemma:covered1}
$\hat{Z}$ contains all the nodes in $W_1 \cup W_2 \setminus \hat{R}$ that are $\omega_1$-covered by $\hat{R}$, with probability of 1.
\end{lemma}
\begin{lemma}[Quality of $\hat{Z}$: Part 2]
\label{lemma:covered2}
Further, suppose all nodes in $\hat{R}$ are good, then all nodes in $\hat{Z}$ are $\omega_2$-covered by $\hat{R}$, $\forall \omega_2 \geq (36 \frac{\alpha-1}{\alpha-2})^{\frac{1}{\alpha-2}} \omega_1$.
\end{lemma}

\noindent
\textbf{Half Duplex Communication}.
%\label{sec:halfduplex}
Now, we assume that nodes are in the half duplex mode, so that they cannot perform transmission and reception/sensing at the same time.
%We show the modification of Algorithm~\ref{alg:ex-ruling} to account for half duplex communication and give a brief proof of Lemma~\ref{theorem:halfduplex} below in Appendix~\ref{append:halfduplex}.
In Algorithm~\ref{alg:ex-ruling},
Lines~\ref{algline:loop3} through~\ref{algline:coor-end} make use of the full duplex capability, such that
Lemma~\ref{lemma:spaced-full} is true.
To account for the case of half duplex communication,
if we replace the one-slot deterministic full duplex mechanism (Lines~\ref{algline:loop3} through~\ref{algline:coor-end}) with a randomized $O(\log n)$-time loop ---
illustrated by the following lines of pseudo code
--- we have Lemma~\ref{lemma:spaced} for half duplex communication as the counterpart of Lemma~\ref{lemma:spaced-full} for full duplex.
The cost incurred includes (1) the increase in the total running time to obtain an $(\omega_1,\omega_2)$-ruling by $O(\log n)$, and (2) a weakened statement in Lemma~\ref{lemma:spaced} compared to Lemma~\ref{lemma:spaced-full}.

\begin{algorithm}[htbp]
\SetNoFillComment
\SetKwComment{Comment}{/* }{ */}
\NoCaptionOfAlgo
\For(\tcc*[h]{resolving half duplex communication}){$j = 1$ to $C_5 \log n$} { %\label{algline:loop3}
	\Comment{in each slot}
    \lIf{$U(v) = 1$} {
        $v$ transmits with prob. $1/2$\;
		\If{$v$ does not transmit} {
			$v$ senses, $I(v) \gets$ the power $v$ receives in this slot\;
			\lIf{$I(v) > Thres(\omega_1)$} { \label{algline:turnoff}
				$U(v) \gets 0$\Comment*[l]{stops}
			}
		}
	}
}
\caption{In replacement of Lines~\ref{algline:loop3} through~\ref{algline:coor-end} in Algorithm~\ref{alg:ex-ruling}
for using half duplex radios.}
\end{algorithm}

%We plug the following pseudo-code piece into Algorithm~\ref{alg:ex-ruling} in replacement of Lines~\ref{algline:loop3} through~\ref{algline:coor-end}, we obtain Lemma~\ref{lemma:spaced} below and present the proof in Appendix~\ref{append:halfduplex}.
%The proof of Lemma~\ref{lemma:spaced} is in Appendix~\ref{append:ruling}.

\begin{lemma}[Quality of $\hat{R}$: Half Duplex Mode]
\label{lemma:spaced}
All nodes in $\hat{R}$ are good, \whp
%with probability of at least $1 - \frac{1}{n^{C_{10}}}$, for some positive constant $C_{10}$.
\end{lemma}

Since 
Lemmas~\ref{lemma:runtime},~\ref{lemma:covered1} and~\ref{lemma:covered2} remain valid,
we obtain the following theorem for the half duplex case.

\begin{theorem}[Half Duplex]
\label{theorem:halfduplex}
There exists a modified version of $\operatorname{ConstructR}(\omega_1,\omega_2,W_1,W_2,b_{max})$ 
for the half duplex case, such that
it finishes in $O(\log^2 m \log b_{max})$ time and
by the end of the algorithm:
\begin{inparaenum} [(1)]
\item $\hat{R}$ forms an $(\omega_1, \omega_2)$-ruling of $W_1$, and
\item $\hat{Z} \cap W_1 = W_1 \setminus \hat{R}$ and $\{v \in W_2 : v \text{ is } \omega_1\text{-covered by } \hat{R}\} \subseteq \hat{Z} \cap W_2 \subseteq \{v \in W_2 : v \text{ is } \omega_2\text{-covered by } \hat{R}\}$, \whp
\end{inparaenum}
%with probability of at least $1 - \frac{1}{n^{C_{10}}}$, for some positive constant $C_{10}$.
\end{theorem}

\section{Distributed Implementation with Non-adaptive Uniform Transmission Power for Scheduling}
\label{sec:dist-impl}

Putting everything together, we present in this section the distributed implementation of Algorithm~\ref{alg:one-shot}
when restricted to using one uniform power level for scheduling.
\begin{theorem}[Performance]
\label{theorem:distributed-impl-nonadaptive}
Our distributed implementation of Algorithm~\ref{alg:one-shot} with non-adaptive uniform transmission power
has the following properties:
\begin{compactenum} [(1)]
\item in half duplex mode, it terminates in $O(g(L) \log^3 m)$ time,
\item in full duplex mode, it terminates in $O(g(L) \log^2 m)$ time, and
\item in both modes, it produces a constant-approximate solution to \problemu.
\end{compactenum}
\end{theorem}

For the $i$th phase of Algorithm~\ref{alg:one-shot}, we present the distributed implementation that works even when there is only one fixed power level available.
We assign $\gamma_2$ a constant value $\geq (36 \frac{\alpha-1}{\alpha-2})^{\alpha-2} \gamma_1$, and
let $\omega_1=\gamma_1 d_i$, and $\omega_2 =\gamma_2 d_i$.
The distributed implementation goes as follows.

\begin{algorithm}[htbp]
\SetNoFillComment
\SetKwComment{Comment}{/* }{ */}
\SetKwInOut{InputP}{input}
\SetKwInOut{OutputP}{output}

\InputP{Link sets $Y,Y',S$}
\OutputP{$Y^a = \{x(l): l \in Y, A(S,l) \leq \psi (1 - (\frac{\phi}{\beta (1+\phi)})^{1/\alpha})^{\alpha}\}$, \\
$Y^b = \{x(l): l \in Y', A(S,l) \leq \psi (1 - (\frac{\phi}{\beta (1+\phi)})^{1/\alpha})^{\alpha}\}$, \\
$\overline{Y^a} = \{x(l): l \in Y, A(S,l) > \psi (1 - (\frac{\phi}{\beta (1+\phi)})^{1/\alpha})^{\alpha}\}$, \\
$\overline{Y^b} = \{x(l): l \in Y', A(S,l) > \psi (1 - (\frac{\phi}{\beta (1+\phi)})^{1/\alpha})^{\alpha}\}$}
%\OutputP{$\hat{Z}'$: set of all infeasible receiver nodes of links in $Y_1 \cup Y_2$ under the interference from $S$}

%\Comment{$\hat{R}'(r(l)),\hat{Z}'(r(l))$: indicators at $r(l)$ of \\ whether link $l$ belongs to $\hat{R}'$ or $\hat{Z}'$}
%$\hat{R}'(r(l)) \gets \text{false}, \hat{Z}'(r(l)) \gets \text{false}$\;
\Comment{\fontsize{8.5}{10.5pt}\selectfont in 1 time slot:\normalsize}
\lIf{$l \in S$} {
	$x(l)$ transmits\;
}
\ElseIf{$l \in Y \cup Y'$} {
	$x(l)$ senses, $SP(X(S),x(l)) \gets$ the power $x(l)$ receives\; \label{algline:I-l}
    \uIf{$SP(X(S),x(l)) \leq Thres\big((\frac{\beta (1 - \frac{d(l)}{P/(\beta N)})}{\psi (1 - (\frac{\phi}{\beta (1+\phi)})^{1/\alpha})^{\alpha}})^{1/\alpha} d(l)\big)$}{
    	\lIf{$l \in Y$}{$x(l)$ joins $Y^a$};
    	\lIf{$l \in Y'$}{$x(l)$ joins $Y^b$}\;
    }
    \Else{
    	\lIf{$l \in Y$}{$x(l)$ joins $\overline{Y^a}$};
    	\lIf{$l \in Y'$}{$x(l)$ joins $\overline{Y^b}$}\;
    }
}
\Return $Y^a,Y^b,\overline{Y^a},\overline{Y^b}$\;
\caption{$\operatorname{CheckA}(Y,Y',S)$: Distributed algorithm for checking affectance.\label{alg:check-affectance}}
\end{algorithm}

%\subsection{Distributed Implementation: 1st Step}
\noindent
\textbf{Distributed Implementation: 1st Step}:
With Algorithm~\ref{alg:check-affectance}, we run
$\operatorname{CheckA}(J_i, J^{>}_i, \cup_{j < i} S_j)$ to implement the 1st step for phase $i$ in Algorithm~\ref{alg:one-shot}.
$\forall l \in J_i \cup J^{>}_i$, on Line~\ref{algline:I-l} of Algorithm~\ref{alg:check-affectance}, we get $SP(X(S),x(l)) = \sum_{l' \in \cup_{j < i} S_j} \frac{P}{d^{\alpha}(l',\overleftarrow{l})} + N$.
Then, since 
\begin{align*}
& SP(X(S),x(l)) \\
\leq & Thres\big((\frac{\beta (1 - \frac{d(l)}{P/(\beta N)})}{\psi (1 - (\frac{\phi}{\beta (1+\phi)})^{1/\alpha})^{\alpha}})^{1/\alpha} d(l)\big)
\end{align*}
is equivalent to 
\begin{align*}
& A(\cup_{j < i} S_j, \overleftarrow{l}) 
= \frac{\beta}{1 - \frac{d^{\alpha}(l)}{P / (\beta N)}} \sum_{l' \in \cup_{j < i} S_j} \frac{d^{\alpha}(l)}{d^{\alpha}(l',\overleftarrow{l})} \\
\leq & \psi (1 - (\frac{\phi}{\beta (1+\phi)})^{1/\alpha})^{\alpha},
\end{align*}
the sets of links whose sender nodes are in $Y^a,Y^b,\overline{Y^a},\overline{Y^b}$ correspond to $J^a_i,J^b_i,\overline{J^a_i},\overline{J^b_i}$ in Algorithm~\ref{alg:one-shot} respectively.

%\subsection{Distributed Implementation: 2nd Step}
\noindent
\textbf{Distributed Implementation: 2nd Step}:
Recall that for a link set $L'$, $X(L')$ is the set of all sender nodes.
To implement the 2nd step for phase $i$ in Algorithm~\ref{alg:one-shot}.
we feed $b_{max} = m$ to Algorithm~\ref{alg:ex-ruling} and run
$\operatorname{ConstructR}(\omega_1,\omega_2,X(J^a_i),X(J^b_i),m)$.
Thus, we obtain an $(\omega_1,\omega_2)$-ruling $\hat{R}$ of $X(J^a_i)$ and $\hat{Z}$ that complements $\hat{R}$
in $O(\log^3 m)$ time for half duplex and $O(\log^2 m)$ time for full duplex,
Then, the sets of links whose sender nodes are in $\hat{R},\hat{Z}$ respectively correspond to $X(J^r_i), X(J^z_i)$ in Algorithm~\ref{alg:one-shot}.

%\subsection{Distributed Implementation: 3rd Step}
\noindent
\textbf{Distributed Implementation: 3rd Step}:
The 3rd step of Algorithm~\ref{alg:one-shot} means all the links in class $L_i$ and those longer links removed in the 1st step exit Algorithm~\ref{alg:one-shot}. Because our algorithm is sender based, the corresponding links will quit upon the decision of their sender nodes in the 1st and the 2nd steps.

\begin{comment}
\subsection{Correctness of Distributed Implementation}
\label{sec:correctness-dist-impl}
By Theorem~\ref{theorem:alg-ex-ruling} and Lemma~\ref{lemma:alg-step2}, we claim that our distributed implementation correctly follows Algorithm~\ref{alg:one-shot} at each phase and step and hence produces the same result.
\end{comment}

\section{Distributed Implementation with Adaptive Transmission Power for Scheduling}
\label{sec:dist-impl-var-p}

In this section, suppose we have multiple power levels at our disposal on each node\footnote{
In this paper we only study \problemu where links in a solution use a uniform power level for data transmission;
this does not necessarily restrict scheduling control to using uniform power. 
The general version of the problem \problem that
explores power control in both scheduling and data transmission
in a distributed setting is a hard problem and remains open.};
we present how this aids the distributed implementation of Algorithm~\ref{alg:one-shot}.
\begin{theorem}[Performance]
%\label{theorem:distributed-impl-nonadaptive}
\label{theorem:distributed-impl-adaptive}
Our distributed implementation of Algorithm~\ref{alg:one-shot} with adaptive transmission power
has the following properties:
\begin{compactenum} [(1)]
\item in half duplex mode, it terminates in $O(g(L) \log^2 m)$ time,
\item in full duplex mode, it terminates in $O(g(L) \log m)$ time, and
\item in both modes, it produces a constant-approximate solution to \problemu.
\end{compactenum}
\end{theorem}

Again, note that these adaptive power levels are only for scheduling in the control phase; for data transmission in the resulting independent set $S$ we still use one uniform power level. Specifically, we require that
\begin{inparaenum}[(1)]
\item nodes have access to a set of $\Theta(g(L))$ power levels; and
\item for each $i \in [1, g(L)]$, there exists a power level $P_i$ to use such that $(\frac{P_i}{\beta N})^{1/\alpha} = \gamma_3 d_i$, where $\gamma_3$ is a constant.
\end{inparaenum}

We present a second method to implement the 2nd step of each phase in Algorithm~\ref{alg:one-shot}, reducing the running time by one logarithmic factor,
%The idea is to modify the distributed implementation of the 2nd step of Algorithm~\ref{alg:one-shot},
by
\begin{inparaenum} [(1)]
\item
performing a preprocessing to reduce $b_{max}$ to some constant $C_9$ in $O(\log m)$ time,
\item
running Algorithm~\ref{alg:ex-ruling} with the constant $C_9$ in $O(\log^2 m)$ time with half duplex radios and $O(\log m)$ time with full duplex radios, and
\item
performing a postprocessing to obtain the sets of links required as a result of 2nd step of Algorithm~\ref{alg:one-shot} in one slot.
\end{inparaenum}

We introduce a new constant $\gamma_4 \geq (36 \frac{\alpha-1}{\alpha-2})^{\alpha-2} \gamma_1$, and
assign $\gamma_2$ a constant value $\geq \gamma_3 + \gamma_4$.
For the $i$th phase of Algorithm~\ref{alg:one-shot},
let $\omega_1 =\gamma_1 d_i$, $\omega_2 =\gamma_2 d_i$,
$\omega_3 =\gamma_3 d_i$ and $\omega_4 =\gamma_4 d_i$.
%We only use various power levels for the control phase; we still use uniform power assignment for data transmission in $S$.
%The total running time now becomes $O(g(L) \log^2 m)$ with half duplex radios and $O(g(L) \log m)$ with full duplex radios.

We reuse the implementation for the 1st and the 3rd steps from the previous section.
We implement the 2nd step of each phase in Algorithm~\ref{alg:one-shot} with the following three sub-steps.

\subsection{Preprocessing: Constant Density Dominating Set}

Scheideler, Richa, and Santi in \cite{scheideler+:mobihoc08} propose a distributed protocol to construct a \emph{constant density dominating set} of nodes under uniform power assignment within $O(\log m)$ slots. They define $Dom(W, P_t)$ as a dominating set of a node set $W$ with transmission power of $P_t$ on each node, such that $W$ is $d_t$-covered by $Dom(W, P_t)$, where $d_t$ is the transmission range under $P_t$. Then, by "constant density", they mean that $Dom(W, P_t)$ is a $O(1)$-approximation of the minimum dominating set of $W$, such that within the transmission range $d_t$ of each node in $W$ there are at most a constant number $C_9$ of nodes chosen by $Dom(W, P_t)$.

At phase $i$ of Algorithm~\ref{alg:one-shot},
after the 1st step of checking affectance,
we execute the protocol on the node set $X(J^a_i)$ with power $P_i$ which corresponds to a transmission range of $\omega_3$,
and thus obtain a constant density dominating set $Dom(X(J^a_i), P_i)$ out of $X(J^a_i)$.
$Dom(X(J^a_i), P_i)$ has the following properties:
\begin{compactenum}[(1)]
\item $Dom(X(J^a_i), P_i) \in X(J^a_i)$;
\item \textbf{dominating set}: all the node in $X(J^a_i)$ $\omega_3$-covered by $Dom(X(J^a_i), P_i)$; and
\item \textbf{constant density}: $\forall v \in X(J^a_i)$, $1 \leq |B(v, \omega_3) \cap Dom(X(J^a_i), P_i)| \leq C_9$, where $C_9$ is a constant.
\end{compactenum}

\subsection{Construction of Ruling $X(J^r_i)$}

$\operatorname{ConstructR}(\omega_1,\omega_4,Dom(X(J^a_i), P_i),X(J^b_i),C_9)$
produces $\hat{R}$ as an $(\omega_1,\omega_4)$-ruling of $Dom(X(J^a_i), P_i)$, and $\hat{Z}$ such that
\begin{compactenum} [(1)]
\item
$\hat{Z} \cap X(J^a_i) = Dom(X(J^a_i), P_i) \setminus \hat{R}$;
\item
$\hat{Z} \cap X(J^b_i) \supseteq \{v \in X(J^b_i) : \, v \text{ is } \omega_1\text{-covered by } \hat{R}\}$; and
\item
$\hat{Z} \cap X(J^b_i) \subseteq \{v \in X(J^b_i) : \, v \text{ is } \omega_4\text{-covered by } \hat{R}\}$, \ie, $\hat{Z}$ is $\omega_4$-covered by $\hat{R}$.
\end{compactenum}

We argue that $\hat{R}$ is an $(\omega_1,\omega_2)$-ruling of $X(J^a_i)$ due to the following two properties:
\begin{inparaenum} [(1)]
\item $\hat{R} \subseteq X(J^a_i)$ and any two nodes in $\hat{R}$ are $\omega_1$-separated, and
\item $X(J^a_i)$ is $\omega_2$-covered by $\hat{R}$.
\end{inparaenum}
Property (2) can be deduced from the facts below:
\begin{inparaenum} [(1)]
\item $X(J^a_i)$ is $\omega_3$-covered by $Dom(X(J^a_i), P_i)$ due to the preprocessing step,
\item $Dom(X(J^a_i), P_i)$ is $\omega_4$-covered by $\hat{R}$ as a result of $\operatorname{ConstructR}(\omega_1,\omega_4,Dom(X(J^a_i), P_i),X(J^b_i),C_9)$, and
\item $\omega_4 + \omega_3 \leq \omega_2$ by our construction.
\end{inparaenum}
Therefore, $\hat{R}$ corresponds to $X(J^r_i)$ in the 2nd step of Algorithm~\ref{alg:one-shot}.

\subsection{Postprocessing: Accounting for $X(J^z_i)$}

Construct $\hat{Z}' \triangleq \hat{Z} \cup \left(X(J^a_i) \setminus \hat{R}\right)$; the following is true:
\begin{compactenum} [(1)]
\item
$\hat{Z}' \cap X(J^a_i) = X(J^a_i) \setminus \hat{R}$;
\item
$\hat{Z}' \cap X(J^b_i) \supseteq \{v \in X(J^b_i) : \, v \text{ is } \omega_1\text{-covered by } \hat{R}\}$; and
\item
$\hat{Z}' \cap X(J^b_i) \subseteq \{v \in X(J^b_i) : \, v \text{ is } \omega_2\text{-covered by } \hat{R}\}$, \ie, $\hat{Z}'$ is $\omega_2$-covered by $\hat{R}$.
\end{compactenum}
Therefore, $\hat{Z}'$ corresponds to $X(J^z_i)$ in the 2nd step of Algorithm~\ref{alg:one-shot}.

%=========================================================================
%  Conclusion
%=========================================================================
\section{Conclusion}
\label{sec:conclusion}

In this paper, we present the first set of fast distributed
algorithms in the SINR model with a constant factor approximation guarantee
for \problemu.
We extensively study the problem 
by accounting for the cases of half/full duplex and non-adaptive/adaptive power
availability for scheduling. 
The non-local nature of this model and the asymmetry between senders
and receivers makes this model very challenging to study. Our algorithm is
randomized and crucially relies on physical carrier sensing for the distributed
communication steps, without any additional assumptions. Our main technique of
distributed computation of a ruling is likely to be useful in
the design of other distributed algorithms in the SINR model.
\bibliographystyle{IEEEtran}
\bibliography{ref}

%=========================================================================
%  Appendix
%=========================================================================

\appendices

\section{Appendix to Section~\ref{sec:dist-algr-highlevel}}
\label{append:highlevel}

\subsection{Proof of Lemma~\ref{lemma:correctness}}
\label{append:correctness}

\iffalse
\begin{lemma}[Correctness]
\label{lemma:correctness}
The algorithm results in an independent set $S$.
%$RI(S \setminus \{l\}, l) \leq 1$ for each $l \in S$.
\end{lemma}
\fi

The statement of Lemma~\ref{lemma:correctness} is equivalent to that $\forall l \in S, A(S \setminus \{l\}, l) \leq 1$.
Let $l$ be an arbitrary link in $S$, and w.l.o.g., we assume $l \in S_i$, and thus $l \in L_i$.
In each phase $j<i$, because $\overline{J^a_j} \cup \overline{J^b_j} \cup J^r_j \cup J^z_j \supseteq J_j$,
all the links in $\cup_{j<i} L_j$ have been removed from $J$ at the end of phase $j$.
Due to the 2nd step,
$A(\cup_{j < i} S_j, \overleftarrow{l}) \leq \psi (1 - (\frac{\phi}{\beta (1+\phi)})^{1/\alpha})^{\alpha}$.
First, we show that
$A(\cup_{j < i} S_j, l) \leq \psi$.

For any link $l' \in \cup_{j < i} S_j$,
\[
A(l',\overleftarrow{l}) = \frac{\beta}{1-\frac{d^{\alpha}(\overleftarrow{l})}{P / (\beta N)}} \frac{d^{\alpha}(\overleftarrow{l})}{d^{\alpha}(l',\overleftarrow{l})}
\leq A(\cup_{j < i} S_j, \overleftarrow{l}) < 1.
\]

Hence, $d(x(l'), x(l)) = d(l',\overleftarrow{l}) \geq \big(\frac{\beta d^{\alpha}(l)}{1-\frac{d^{\alpha}(l)}{P / (\beta N)}}\big)^{1/\alpha} \geq (\frac{\beta (1+\phi)}{\phi})^{1/\alpha} d(l)$,
implying $\frac{d(l',l)}{d(l',\overleftarrow{l})} = \frac{d(x(l'), r(l))}{d(x(l'), x(l))} \geq \frac{d(x(l'), x(l)) - d(l)}{d(x(l'), x(l))} \geq 1 - (\frac{\phi}{\beta (1+\phi)})^{1/\alpha}$.
By referring to the definition of affectance --- $A(\cup_{j < i} S_j, l) = \sum_{l' \in \cup_{j < i} S_j} \frac{\beta}{1-\frac{d^{\alpha}(l)}{P / (\beta N)}} \frac{d^{\alpha}(l)}{d^{\alpha}(l',l)}$
and $A(\cup_{j < i} S_j, \overleftarrow{l}) = \sum_{l' \in \cup_{j < i} S_j} \frac{\beta}{1-\frac{d^{\alpha}(\overleftarrow{l})}{P / (\beta N)}} \frac{d^{\alpha}(\overleftarrow{l})}{d^{\alpha}(l',\overleftarrow{l})}$ --- we obtain
\[
A(\cup_{j < i} S_j, l)
\leq \big(\frac{1}{1 - (\frac{\phi}{\beta (1+\phi)})^{1/\alpha}}\big)^{\alpha} A(\cup_{j < i} S_j, \overleftarrow{l})
= \psi.
\]

Next,
it suffices to show that
$$A(\cup_{j \geq i} S_j \setminus \{l\}, l) \leq 1-\psi.$$

At phase $i$, $\omega_1 = \gamma_1 d_{i}$;
the 1st step in the phase ensures that when link $l$ is added to $S_i$, any link $l' \in \cup_{j \geq i} L_j$ with $d(x(l'), x(l)) < \omega_1$ will not get in $S_i$.
Therefore,
%for link $l$, all the sender nodes of the same or higher link classes have a mutual distance of at least $\omega_1$.
for the set $\cup_{j \geq i} S_j \setminus \{l\}$, we have:
\begin{inparaenum} [(1)]
\item all the nodes in $X(\cup_{j \geq i} S_j \setminus \{l\})$ have a mutual distance of at least $\omega_1$;
\item the distance from any node in $X(\cup_{j \geq i} S_j \setminus \{l\})$ to $r(l)$ is at least $\omega_1 - d(l)$; and
\item $\omega_1 - d(l) > \omega_1/2 > 0$.
\end{inparaenum}
According to Proposition~\ref{prop:sensed-interference}, by using $\gamma_1 = \big(\frac{36 \beta}{1-\psi}\frac{\alpha-1}{(\alpha-2)}\frac{1+\phi}{\phi}\big)^{1/\alpha} + 2$,
$SP\big(X(\cup_{j \geq i} S_j \setminus \{l\}), r(l)\big) \leq \frac{36 (\alpha-1)}{\alpha-2} \frac{P}{(\omega_1 - d(l))^{\alpha}} + N
< \frac{(1-\psi) P}{\beta d^{\alpha}(l)} (1- \frac{d^{\alpha}(l)}{P/(\beta N)}) + N$.
It is easy to verify $A(\cup_{j \geq i} S_j \setminus \{l\}, l) < 1-\psi$.

\begin{proposition}
\label{prop:sensed-interference}
$\forall V' \in V$ and $\forall v \not\in V'$, if
\begin{inparaenum} [(1)]
\item
all the nodes in $V'$ are at least $\rho_1$ away from each other,
%$\forall v,v' \in L', d(v,v') \geq \rho_1$,
\item
the distance between $v$ and any node in $V'$ is at least $\rho_2$,
%$\forall v \in V', d(v,w) \geq \rho_2$, and
\item $\rho_2 > \rho_1 / 2 > 0$,
\end{inparaenum}
then $SP(V',v) < \frac{36 (\alpha-1)}{\alpha-2} \frac{\rho_2^2}{\rho_1^2} \frac{P}{\rho_2^{\alpha}} + N$.
\end{proposition}
\begin{proof}
We bound the sensed power strength by partitioning the plane into concentric rings all centered at $v$, each of width $\rho_2$, via a similar technique to that in \cite{chafekar:mobihoc07, Moscibroda+:Mobihoc06}.
Let $Ring(i)$ denote the $i$th ring (where $i = 1,2,\ldots$), which contains every node $v'$ that satisfies $i \rho_2 \leq d(v', v) < (i+1) \rho_2$; let $V'(i)$ denote the subset of nodes in $V'$ that fall in $Ring(i)$.
We notice the following facts:
\begin{inparaenum} [(1)]
\item
For any two nodes $v,v' \in V'(i)$, two disk centered at $v,v'$ respectively with a radius of $\rho_1 / 2$ are non-overlapping.
\item
For any node $v \in V'(i)$, such a disk is fully contained in an extended ring $Ring'(i)$ of $Ring(i)$, with an
extra width of $\rho_1 / 2$ at each side of $Ring(i)$.
\end{inparaenum}
The area (denoted by $D$) of each of such disks is $D = \pi (\rho_1 / 2)^2$.
The area (denoted by $D(i)$) of $Ring'(i)$ is
\begin{align*}
D(i) &= \pi [((i+1) \rho_2 + \rho_1/2)^2 - (i \rho_2 - \rho_1/2)^2] \\
&\leq 3 \pi (2i+1) \rho_2^2.
\end{align*}

Using $|V'(i)| \leq D(i) / D \leq 12 (2i+1) \rho_2^2 / \rho_1^2$, we obtain
% \leq 12 (2h+1) \rho_2^2 / \rho_1^2$.
\begin{align*}
SI(V',w)
&\leq  \sum_{i = 1}^{\infty} |V'(i)| \frac{P}{(i \rho_2)^{\alpha}} + N \\
&\leq \sum_{i = 1}^{\infty} \frac{12 (2i+1)}{i^{\alpha}} \frac{\rho_2^2}{\rho_1^2} \frac{P}{\rho_2^{\alpha}} + N \\
&\leq  \frac{36 (\alpha-1)}{\alpha-2} \frac{\rho_2^2}{\rho_1^2} \frac{P}{\rho_2^{\alpha}} + N. \qedhere
\end{align*}
\end{proof}

%\subsection{Analysis of Algorithm~\ref{alg:one-shot} for Distributed Maximum Link Scheduling}
\subsection{Proof of Theorem~\ref{theorem:constantS}}
\label{append:constantS}

\iffalse
\begin{theorem}[Approximation Ratio]
\label{theorem:constantS}
$\forall \gamma_1, \gamma_2, \psi > 0$, $|OPT(L)| \leq C_3(\gamma_2, \psi) |S|$ if $\gamma_2 > \gamma_1 > 1$, where $C_3(\gamma_2, \psi) = C_1(\gamma_2) + C_2(\psi (1 - (\frac{\phi}{\beta (1+\phi)})^{1/\alpha})^{\alpha})$.
\end{theorem}

Algorithm~\ref{alg:one-shot} results in a constant approximation ratio presented in Theorem~\ref{theorem:constantS}.
%the expression of the constants mentioned and its proof are in Appendix~\ref{append:constantS}, where
This theorem is backed by Lemma~\ref{lemma:constant-step1} and Lemma~\ref{lemma:constant-step2}. Each of the two lemmas, independent of Algorithm~\ref{alg:one-shot}, stands on its own to provide insights of how an optimum solution is shaped under the SINR model.
Now we present the proof of Theorem~\ref{theorem:constantS} below.
%\textcolor{red}{The basic idea is to show that in each iteration, the number of links we discard is small.}
\fi

For a node $v$, we define $B(v, d)$ as the ball centered at $v$ with a radius of $d$.
With a parameter $\gamma > 1$, we then define a link set $B_{\gamma}^{\geq}(l)$,
such that for a link $l \in L_i$,  $B_{\gamma}^{\geq}(l)$ contains all and only the links in $\{L_j: j \geq i\}$, with their senders in the ball $B(x(l), \gamma d_{i})$;
in other words, $B_{\gamma}^{\geq}(l)$ contains the links with similar or longer lengths, whose senders are $(\gamma d_{i})$-covered by $x(l)$.
For a set $L' \subseteq L$, $B_{\gamma}^{\geq}(L')$ is defined as $\cup_{l \in L'} B_{\gamma}^{\geq}(l)$.
%Note that $B_{\gamma}^{\geq}(L')$ means that the ``covering'' links are in $L'$, while it covers links in $L$.

\begin{lemma}[Spatial Constraint]
\label{lemma:constant-step1}
$\forall \gamma > 1$, $\forall L' \subseteq L$, $\forall l \in L'$, $|OPT(B_{\gamma}^{\geq}(l) \cap L')| \leq C_1(\gamma)$, where $C_1(\gamma) = \frac{(2\gamma+1)^{\alpha}} {\beta}$.
\end{lemma}
\begin{proof}
Let $k$ be an arbitrary link in $OPT(B_{\gamma}^{\geq}(l) \cap L')$.
%W.l.o.g., we assume $k \in L_i$; then $d_i/2 \leq d(k) \leq d_i$.
For any link $k' \in OPT(B_{\gamma}^{\geq}(l) \cap L')$,
$d(k',k) \leq d(x(k'),x(l)) + d(x(l),x(k)) + d(k) \leq (2\gamma+1) d(k)$.
Therefore,
\begin{align*}
1 & \geq A(OPT(B_{\gamma}^{\geq}(l) \cap L') \setminus \{k\}, k) \\
	& = \beta \frac{\sum_{k' \in OPT(B_{\gamma}^{\geq}(l) \cap L') \setminus \{k\}} \frac{d^{\alpha}(k)}{d^{\alpha}(k',k)}}{1 - \frac{d^{\alpha}(k)}{P / (\beta N)} } \\
	& \geq \beta \frac{|OPT(B_{\gamma}^{\geq}(l) \cap L')| \frac{d^{\alpha}(k)}{(2\gamma+1)^{\alpha} d^{\alpha}(k)}}{1 - \frac{d^{\alpha}(k)}{P / (\beta N)} } \\
	 & \geq \beta (2\gamma+1)^{-\alpha} |OPT(B_{\gamma}^{\geq}(l) \cap L')|.
\end{align*}
It follows that $|OPT(B_{\gamma}^{\geq}(l) \cap L')| \leq (2\gamma+1)^{\alpha} / \beta = C_1(\gamma)$.
%\qedhere
\end{proof}

\begin{lemma}[Affectance Constraint]
\label{lemma:constant-step2}
$\forall \psi' > 0$ and $\forall L', L'' \subseteq L$,
if $L' \cap L'' = \emptyset$ and $ A(L', \overleftarrow{l}) > \psi'$ for any link $l \in L''$,
then
$|OPT(L'')| \leq C_2(\psi') |L'|$, where $C_2(\psi') = (\frac{2(\beta b)^{1/\alpha}}{(\beta b)^{1/\alpha} - 1})^{\alpha} / \psi'+1$.
\end{lemma}
\begin{proof}
If $|OPT(L'')| > 0$, we can express it as
$|OPT(L'')| = b |L'| + g$, such that $b$ and $g$ are non-negative integers and $1 \leq g \leq |L'|$.
We create $|L'|$ bins, each of which has a capacity of $b$ links;
we pack the links in $OPT(L'')$ to the bins via a first-fit sweep through the links in $L'$:
\begin{enumerate} [(1)]
\item
%Let $S(S)$ denote a sequence of links in $S$, and
We order the links in $L'$ arbitrarily;
let $l_j$ denote the $j$th link in $L'$, and let $Bin_j$ denote the $j$th bin.
Let set $L^* = OPT(L'')$ initially; then, the sweep proceeds in $|L'|$ rounds.
\item
In the $i$th round (where $i=1,2,\ldots,|L'|$), we pick $b$ links in $L^*$ whose senders are the nearest $b$ nodes to the sender $x(l_i)$ of the $i$th link in $L'$, and we remove those $b$ links from $L^*$ and put them into $Bin_i$.
\end{enumerate}

The completion of the above packing means that in $OPT(L'')$, we have $b |L'|$ links ``near'', and $g \in [1, |L'|]$ links ``far'' from the senders of links in $L'$.
If $b \leq (\frac{2(\beta b)^{1/\alpha}}{(\beta b)^{1/\alpha} - 1})^{\alpha} / \beta$, we are done.
Therefore, for the rest of the proof we assume that $b > (\frac{2(\beta b)^{1/\alpha}}{(\beta b)^{1/\alpha} - 1})^{\alpha} / \beta$,
and we show that $b \leq C_2(\psi') - 1 = (\frac{2(\beta b)^{1/\alpha}}{(\beta b)^{1/\alpha} - 1})^{\alpha} / \psi'$ in this case.
%If $\frac{3^{\alpha}}{\beta} \geq \frac{2^{\alpha}}{\beta(\psi)}$ (depending on $\psi$), there exists a conflict and implies that the assumption that $f > \frac{3^{\alpha}}{\beta}$ is false, and the statement of Lemma~\ref{lemma:constant-step2} holds; otherwise, the statement remains true.

Let $l$ denote a ``far'' link in $OPT(L'')$ that is out of any bins.
We have for any link $k \in Bin_i, d(x(k), x(l_i)) < d(x(l), x(l_i))$,
implying that $d(x(l), x(k)) < 2 d(x(l), x(l_i))$ due to triangle inequality.
%Hence, for all $i$, all the senders of the links in $Bin_i$ fall in a circular region centered at $x(l)$ with radius of $2 d(x(l), x(l_i))$.
Since $Bin_i \cup \{l\} \subseteq OPT(L'')$, we have
\begin{align*}
1 & \geq A(Bin_i, l) = \beta \frac{\sum_{k \in Bin_i} \frac{d^{\alpha}(l)}{d^{\alpha}(k,l)}}{1 - \frac{d^{\alpha}(l)}{P / (\beta N)} } \\
	& \geq \beta \frac{b \big(\frac{d(l)}{2d(x(l), x(l_i)) + d(l)}\big)^{\alpha}}{1 - \frac{d^{\alpha}(l)}{P / (\beta N)} } \\
	& \geq \beta b \big(\frac{1}{2d(x(l), x(l_i)) \big/d(l) + 1}\big)^{\alpha}.
\end{align*}

That leads to that $d(l) \leq \frac{2}{(\beta b)^{1/\alpha} - 1} d(x(l), x(l_i))$.
\begin{comment}
Because $d(l) \geq d(x(l), x(l_i)) - d(r(l), x(l_i))$ (due to triangle inequality) and because we have assumed that $b > \frac{4^{\alpha}}{\beta}$, we obtain
\[
d(x(l), x(l_i)) \leq \frac{(\beta b)^{1/\alpha} - 1}{(\beta b)^{1/\alpha}-3} d(x(l_i),r(l))
\leq 3 d(x(l_i),r(l)).
\]
\end{comment}
Then for any link $k \in Bin_i$,
\begin{equation}
\label{ineq:distance-step2}
\begin{split}
& d(k,l) = d(x(k),r(l)) \\
\leq & d(x(k),x(l_i)) + d(x(l_i),x(l)) + d(x(l), r(l)) \\
\leq & \frac{2(\beta b)^{1/\alpha}}{(\beta b)^{1/\alpha} - 1} d(x(l_i),x(l)).
\end{split}
\end{equation}

Since $\cup_i Bin_i \cup \{l\} \subseteq OPT(L'')$, we have
\begin{align*}
1 & \geq A(\cup_{i=1}^{|L'|} Bin_i, l)
	= \beta \sum_{i=1}^{|L'|} \frac{\sum_{k \in Bin_i} \frac{d^{\alpha}(l)}{d^{\alpha}(k,l)}}{1 - \frac{d^{\alpha}(l)}{P / (\beta N)} } \\
	& \geq \beta \sum_{i=1}^{|L'|} \frac{\frac{|Bin_i| d^{\alpha}(l)}{(\frac{2(\beta b)^{1/\alpha}}{(\beta b)^{1/\alpha} - 1})^{\alpha} d^{\alpha}(x(l_i),x(l))}}{1 - \frac{d^{\alpha}(l)}{P / (\beta N)} }  \;\;\text{(by Ineq.~\eqref{ineq:distance-step2})}\\
	& \geq \frac{b}{(\frac{2(\beta b)^{1/\alpha}}{(\beta b)^{1/\alpha} - 1})^{\alpha}} \frac{\beta \sum_{l' \in L'} \frac{d^{\alpha}(l)}{d^{\alpha}(x(l'),x(l))}}{1 - \frac{d^{\alpha}(l)}{P / (\beta N)} } \\
    & = b (\frac{(\beta b)^{1/\alpha} - 1}{2(\beta b)^{1/\alpha}})^{\alpha} A(L', \overleftarrow{l})  \\
    & \geq b (\frac{(\beta b)^{1/\alpha} - 1}{2(\beta b)^{1/\alpha}})^{\alpha} \psi'.
    %\quad\quad\quad\quad (\text{because } A(L', \overleftarrow{l}) = \beta \frac{\sum_{l' \in L'}  \frac{d^{\alpha}(l)}{d^{\alpha}(x(l'),x(l))}}{1 - \frac{d^{\alpha}(l)}{P / (\beta N)} } \geq \psi')
\end{align*}

The last inequality above holds because $$A(L', \overleftarrow{l}) = \beta \frac{\sum_{l' \in L'}  \frac{d^{\alpha}(l)}{d^{\alpha}(x(l'),x(l))}}{1 - \frac{d^{\alpha}(l)}{P / (\beta N)} } \geq \psi'.$$

Therefore, $b \leq (\frac{2(\beta b)^{1/\alpha}}{(\beta b)^{1/\alpha} - 1})^{\alpha} / \psi'$ and
$|OPT(L'')| \leq (b+1) |L'| \leq C_2(\psi') |L'|$.
\qedhere
\end{proof}

We define $J^a = \cup_i J^a_i, \overline{J^a} = \cup_i \overline{J^a_i}, J^b = \cup_i J^b_i, \overline{J^b} = \cup_i \overline{J^b_i}, J^r = \cup_i J^r_i, J^z = \cup_i J^z_i$.
$\overline{J^a} \cup \overline{J^b}$ contains all the links removed in the 1st step in Algorithm~\ref{alg:one-shot} due to the affectance constraints.
At each phase $i$, $X(J^r_i)$ is an $(\omega_1,\omega_2)$-ruling of $X(J^a_i)$,
and all the nodes in $X(J^z_i)$ are $\omega_2$-covered by $X(J^r_i)$;
we choose all the links in $J^r_i$ to add to $S$,
discard all the links in $\overline{J^a_i} \cup \overline{J^b_i}$ (for failing affectance check),
and also discard all the links in $J^z_i$ (because of their proximity to the chosen links).
%$\overline{J^a_i} \cup \overline{J^b_i} \cup J^r_i \cup J^z_i \suqseteq J_i$.

We have
\begin{align*}
L & = \cup_i (\overline{J^a_i} \cup \overline{J^b_i} \cup J^r_i \cup J^z_i) \\
& \subseteq \cup_i (\overline{J^a_i} \cup \overline{J^b_i}) \cup \cup_i B_{\gamma_2}^{\geq}(J^r_i)\\
& = \overline{J^a} \cup \overline{J^b} \cup B_{\gamma_2}^{\geq}(S).
\end{align*}

Due to Lemma~\ref{lemma:constant-step2},
$|OPT(\overline{J^a} \cup \overline{J^b}| \leq C_2(\psi (1 - (\frac{\phi}{\beta (1+\phi)})^{1/\alpha})^{\alpha}) |S|$;
due to Lemma~\ref{lemma:constant-step1},
$
|OPT(B_{\gamma_2}^{\geq}(S))|
\leq \sum_{l \in S} |OPT(B_{\gamma_2}^{\geq}(l))| \leq C_1(\gamma_2) |S|.
$
Therefore,
\begin{align*}
& |OPT(L)|
%= & |OPT(L) \cap (R \cup R' \cup Z)| + |OPT(L) \cap Z'| \\
\leq |OPT(\overline{J^a} \cup \overline{J^b})| + |OPT(B_{\gamma_2}^{\geq}(S))| \\
\leq & \big(C_1(\gamma_2)+C_2(\psi (1 - (\frac{\phi}{\beta (1+\phi)})^{1/\alpha})^{\alpha})\big) |S|.
%\qquad\qquad\qquad\qedhere
\end{align*}

\section{Appendix to Section~\ref{sec:dist-algr-ruling}}
\label{append:ruling}

Recall that we call an iteration of the outer loop (Line~\ref{algline:loop1}) a phase of the algorithm, and an iteration of the inner loop (Line~\ref{algline:loop2}) a round;
recall that $B(v,d)$ denotes the ball centered at $v$ with a radius of $d$.
%Let $B_{W_1}(v,d)$ denote the set of active nodes in the ball $B(v,d)$;
Let $A_t^{W_1}(v,d)$ denote the set of \emph{active} nodes in set $W_1$ that fall in the ball $B(v,d)$ at time point $t$; we will explicitly point $t$ out whenever we use $A_t^{W_1}(v,d)$.

\subsection{Proof of Lemma~\ref{lemma:runtime}}
\label{append:runtime}

The following definitions are only involved in Lemmas~\ref{lemma:lucky} and~\ref{lemma:success}.
Let $\eta$ be a constant $> \bigl(96 \frac{\alpha-1}{\alpha-2}\bigr)^{-1/\alpha}$.
In one round (which corresponds to one iteration of the inner loop) of Algorithm~\ref{alg:ex-ruling}, let $U \in W_1$ be the set of nodes with $U() = 1$ at line~\ref{algline:inU} in the coordination step.
We say an active node $v \in W_1$ is ``lucky'' in a round, with $t_0$ being the time that the round starts, if and only if
\begin{compactenum} [(1)]
\item
$v \in U$;
\item
$U \cap A_{t_0}^{W_1}(v, \eta\omega_1) = \{v\}$, \ie, $v$ has no nearby active nodes in $U$;
\item
$SP(U \setminus A_{t_0}^{W_1}(v, \eta\omega_1), v) < Thres(\omega_1)$, \ie, total power received from faraway active nodes is small.
\end{compactenum}
In a round if $v$ gets lucky, $U(v)$ will remain 1 till the end of that round, and thus will elect to be included in $\hat{R}$ and will cause all nodes in $A_{t_0}^{W_1}(v,\omega_1)$ to get into $\hat{Z}$.
\begin{lemma}
\label{lemma:lucky}
In a round $i_{in}$ of phase $i_{out}$, with $t_1$ being the time that the round starts, suppose that for each active node $u \in W_1, |A_{t_1}^{W_1}(u,\omega_1)| \leq 2^{\log b_{max} - i_{out} + 1}$ at the beginning of the round,
then the probability for an arbitrary active node in $W_1$ to be lucky in the round is at least $2^{-(\log b_{max} - i_{out} + 3 +  (2\eta + 1)^2)}$.
\end{lemma}
\begin{proof}
%Let $v$ be an active node in $W_1$.
For a round at phase $i_{out}$, we prove the statement in the following four steps.
\begin{asparaenum} [(1)]
\item
First,
for any active node $v \in W_1$, the probability for $v$ to be in $U$ is
$Prob\bigl(U(v)=1\bigr) = 2^{-(\log b_{max} - i_{out} + 2)}$.

\item
Second, for any active node $v \in W_1$, due to the packing property we upper-bound the size of $A_{t_1}^{W_1}(v,\eta\omega_1)$ as
\begin{align*}
\bigl|A_{t_1}^{W_1}(v,\eta\omega_1)\bigr| &\leq \frac{\pi (\eta \omega_1 + \omega_1 / 2)^2}{\pi (\omega_1 / 2)^2} \max_{v'}\bigl|A_{t_1}^{W_1}(v',\omega_1)\bigr| \\
&\leq (2\eta + 1)^2 2^{\log b_{max} - i_{out} + 1}.
\end{align*}
Then,
because $\log b_{max} + 1 \geq i_{out}$, $2^{\log b_{max} - i_{out} + 2} \geq 2$.
the probability for all nodes other than $v$ in $A_{t_1}^{W_1}(v,\eta\omega_1)$ to not appear in $U$ (\ie. to remain silent) is
\begin{align*}
& Prob\bigl(A_{t_1}^{W_1}(v,\eta\omega_1) \cap U \setminus \{v\} = \emptyset \bigr) \\
\geq & \prod_{u \in A_{t_1}^{W_1}(v,\eta\omega_1)} \bigl(1 - Prob(U(u)=1)\bigr) \\
\geq & \bigl(1 - 2^{-(\log b_{max} - i_{out} + 2)}\bigr)^{|A_{t_1}^{W_1}(v,\eta\omega_1)|} \\
\geq & \bigl(1 - 2^{-(\log b_{max} - i_{out} + 2)}\bigr)^{(2\eta + 1)^2 2^{\log b_{max} - i_{out} + 1}} \\
\geq & (1 / 4)^{(2\eta + 1)^2 / 2} = 2^{-(2\eta + 1)^2}.
\end{align*}

\item
Third, for any active node $v \in W_1$,
we lower-bound the probability that $v$'s received power $SP(U \setminus A_{t_1}^{W_1}(v,\eta \omega_1), v)$ from outside of the ball $B(v,\eta \omega_1)$ is ``low'' --- \ie, below $Thres(\omega_1)$ --- by
\begin{inparaenum} [(i)]
\item
partitioning the plane into concentric rings via a similar technique to that in \cite{chafekar:mobihoc07, Moscibroda+:Mobihoc06}, and
\item
referring to an $(\eta \omega_1,\eta \omega_1)$-ruling and determining the expected number of nodes in each a ring that appear in $U$,
so that we can bound the power received from all the nodes in the rings outside of $B(v, \eta \omega_1)$.
\end{inparaenum}

We partition the plane into rings all centered at $v$, each of width $\eta \omega_1$.
Let $Ring(h)$ denote the $h$th ring, which contains every node $v'$ that satisfies $h \eta \omega_1 \leq d(v', r(l)) < (h+1) \eta \omega_1$, for each $h = 1,2,\ldots$; let $Ring^a(h)$ denote the set of active nodes in $Ring(h)$.
When $h = 0$, $Ring(0)$ corresponds to the ball $B(v, \eta \omega_1)$.
Let $R(h)$ denote an $(\eta \omega_1,\eta \omega_1)$-ruling of $Ring^a(h)$.
%The existence of such a ruling $R(h)$ is obvious.
Then by noticing that
\begin{inparaenum} [(i)]
\item $Ring^a(h) \subseteq \cup_{v' \in R(h)} A_{t_1}^{W_1}(v',\eta \omega_1)$, and
\item for any two nodes $v',u' \in R(h), d(v', u') > \eta \omega_1$,
\end{inparaenum}
we have
\begin{equation}
\label{eqn:expected-num-tx}
\begin{split}
& \mathbb{E}\left\{\bigl|U \cap Ring^a(h)\bigr|\right\} \\
= & \displaystyle \sum_{v' \in Ring^a(h)} \mathbb{E}\{U(v')\}  \\
\leq & \sum_{v' \in R(h)} \sum_{u' \in A_{t_1}^{W_1}(v',\eta \omega_1)} \mathbb{E}\{U(u')\}  \\
= & \sum_{v' \in R(h)} \bigl|A_{t_1}^{W_1}(v',\eta \omega_1)\bigr| \frac{b_{max}}{2^{i_{out} - 2}} \\
\leq & 2 \bigl|R(h)\bigr|.
\end{split}
\end{equation}

To bound the cardinality of $R(h)$, we use the following facts:
\begin{inparaenum} [(i)]
\item
for any two nodes $v',u' \in R(h)$, two disk centered at $v',u'$ respectively with a radius of $\eta \omega_1 / 2$ are non-overlapping; and
\item
For any node $v' \in Ring^a(h)$, such a disk is fully contained in an extended ring $Ring'(h)$ of $Ring(h)$, with an
extra width of $\eta \omega_1 / 2$ at each side of $Ring(h)$.
\end{inparaenum}
\begin{comment}
%%%%%%%% to save space
The area (denoted by $A$) of each of such disks is $A = \pi (\eta \omega_1 / 2)^2$.
The area (denoted by $A(h)$) of $Ring'(h)$ is
\begin{align*}
A(h) & = \pi [((h+1) \eta \omega_1 + \eta \omega_1/2)^2 - (h \eta \omega_1 - \eta \omega_1/2)^2] \\
     & = 2 \pi (2h+1) (\eta \omega_1)^2.
\end{align*}
Then, $\bigl|R(h)\bigr| \leq A(h) / A = 8 (2h+1)$, and Inequality~\eqref{eqn:expected-num-tx} yields
$
\mathbb{E}\{|U \cap Ring^a(h)|\} \leq 4 (2h+1).
$
\end{comment}
Then, by referring to the ratio between the areas of $Ring'(h)$ and a disk,
we have $\bigl|R(h)\bigr| \leq 8 (2h+1)$; Inequality~\eqref{eqn:expected-num-tx} yields
$
\mathbb{E}\left\{\bigl|U \cap Ring^a(h)\bigr|\right\} \leq 2^4 (2h+1).
$

Therefore, $v$'s received power from outside of the ball $B(v,\eta \omega_1)$ is
\begin{align*}
& \mathbb{E}\{SP(U \setminus A_{t_1}^{W_1}(v,\eta \omega_1),v)\}\\
%= & \mathbb{E}\left\{\sum_{v' \in U \setminus A_{t_1}^{W_1}(v,\eta \omega_1)} U(v') \frac{P}{d^{\alpha}(v',v)} + N\right\}\\
= & \mathbb{E}\Big\{\sum_{h = 1}^{\infty} \sum_{v' \in U(h)} \frac{P'}{d^{\alpha}(v',v)} + N\Big\} \\
\leq & \sum_{h = 1}^{\infty} \mathbb{E}\left\{\bigl|U \cap Ring^a(h)\bigr|\right\} \frac{P'}{(h \eta \omega_1)^{\alpha}} + N \\
\leq & \sum_{h = 1}^{\infty} \frac{2^4 (2h+1)}{h^{\alpha}} \frac{P'}{(\eta \omega_1)^{\alpha}} + N \\
\leq & \frac{48}{\eta^{\alpha}} \frac{\alpha-1}{\alpha-2} \frac{P'}{\omega_1^{\alpha}} + N
\leq  Thres(\omega_1) / 2.
\end{align*}

According to Markov's Inequality,
$$Prob\bigl(SP(U \setminus A_{t_1}^{W_1}(v,\omega_2),v) \geq Thres(\omega_1)\bigr) \leq 1 / 2,$$
implying
$$Prob\bigl(SP(U \setminus A_{t_1}^{W_1}(v,\omega_2),v) < Thres(\omega_1)\bigr) \geq 1 / 2.$$

\item
Finally, combining the above three, the probability that $v$ is lucky is at least $2^{-(\log b_{max} - i_{out} + 3 +  (2\eta + 1)^2)}$.
\qedhere
\end{asparaenum}
\end{proof}

\begin{lemma}
\label{lemma:success}
In a round $i_{in}$ of phase $i_{out}$, with $t_1$ being the time that the round starts, suppose that for each active node $u \in W_1, \bigl|A_{t_1}^{W_1}(u,\omega_1)\bigr| \leq 2^{\log b_{max} - i_{out}+1}$; then,
for an arbitrary active node $v \in W_1$ with $\bigl|A_{t_1}^{W_1}(v,\omega_1)\bigr| \geq 2^{\log b_{max} - i_{out}}$,
the probability that $v$ becomes inactive by the end of the round is at least a constant $C_6$, where $0 < C_6=2^{-(3 + (2\eta + 1)^2)} < 1$.
\end{lemma}

\begin{proof}
In a round of phase $i_{out}$, a sufficient condition for $v$ to be inactive by the end of the round is that either $v$ or any node in $A_{t_1}^{W_1}(v, \omega_1)$ enters $\hat{R}$, such that $v$ either enters $\hat{R}$ or $\hat{Z}$ and exits the algorithm.
Further, that either $v$ or any node in $A_{t_1}^{W_1}(v, \omega_1)$ gets lucky satisfies this condition.
Therefore, the probability for $v$ to become inactive in the round is at least
\begin{align*}
& \sum_{v' \in A_{t_1}^{W_1}(v, \omega_1)} Prob(v' \text{ is lucky}) \\
\geq & \bigl|A_{t_1}^{W_1}(v, \omega_1)\bigr| 2^{-\left(\log b_{max} - i_{out} + 3 + (2\eta + 1)^2\right)} \\
\geq & 2^{-\left(3 + (2\eta + 1)^2\right)}.
\qedhere
\end{align*}
\end{proof}

\begin{lemma}
\label{lemma:degree}
Let $E_{i_{out}}$ denote the event that at the end of the phase $i_{out}$, $\bigl|A_{t_1}^{W_1}(u, \omega_1)\bigr| \leq b_{max} / 2^{i_{out}}$ for every active node $u \in W_1$, where $t_1$ is the time that the last round in $i_{out}$ ends.
$\forall i_{out}$, $Prob(E_{i_{out}}) \geq 1 - i_{out} / n^{C_7}$, for some positive constant $C_7$.
\end{lemma}

\begin{proof}[proof by induction]
At the end of the phase $i_{out} = 0$, this is trivial.
Suppose that for $i_{out} = i \geq 1$ the statement is true; we show that it still holds for $i_{out} = i + 1$.
At the beginning of phase $i+1$, if we already have event $E_{i+1}$, we are done; otherwise, let $v$ denote an active node such that $\bigl|A_{t_2}^{W_1}(v, \omega_1)\bigr| > \frac{b_{max}}{2^{i+1}}$, where $t_2$ is the time that the first round in $i+1$ starts.
We show $Prob(E_{i+1}) \geq 1 - \frac{i+1}{n^{C_7}}$ as below:
\begin{asparaenum}[(1)]
\item We choose a constant $C_4 \geq \frac{C_7+1}{\log \frac{1}{1 - C_6}}$ for the inner loop at Line~\ref{algline:loop2} of Algorithm~\ref{alg:ex-ruling}.
\item Under the induction assumption for phase $i$,
%that after phase $i$, $\bigl|A_{t_2}^{W_1}(u, \omega_1)\bigr| \leq \frac{b_{max}}{2^{i}}$ for every active node $u \in W_1$,
we have:
for any round $i_{in}$ during phase $i+1$, with $t_3$ being the time that this round starts,
if $\bigl|A_{t_3}^{W_1}(v, \omega_1)\bigr| \leq \frac{b_{max}}{2^{i+1}}$, we are done;
otherwise, as long as $\frac{b_{max}}{2^{i}} \geq \bigl|A_{t_3}^{W_1}(v, \omega_1)\bigr| \geq \frac{b_{max}}{2^{i+1}}$, the probability for $v$ to turn inactive during the round is at least $C_6$ according to Lemma~\ref{lemma:success}.
Then, the probability for $v$ to become inactive by the end of phase $i+1$ (consisting of $C_4 \log n$ rounds) is at least $1 - (1 - C_6)^{C_4 \log n} \geq 1 - 1/n^{C_7+1}$.
\item By considering both the conditional probability and the fact that there are at most $n$ such nodes as $v$,
\begin{align*}
& Prob(E_{i+1}) \\
\geq & Prob(E_{i+1} \bigm| E_{i}) Prob(E_{i}) \\
\geq & (1 - \frac{n}{n^{C_7+1}})(1 - \frac{i}{n^{C_7}}) \geq 1 - \frac{i+1}{n^{C_7}}.
\qedhere
\end{align*}
\end{asparaenum}
\end{proof}

At the end of phase $i_{out} = \log b_{max} + 1$, with $t_4$ being the time that the last round in phase $i_{out}$ ends, we have that with probability of at least $1 - \frac{\log b_{max} + 1}{n^{C_7}}$, for every active node $u \in W_1$,
$|A_{t_4}^{W_1}(u, \omega_1)| \leq 1/2 < 1$.
%We can show that in the next phase, by a similar proof to that of Lemma~\ref{lemma:success}, all these active nodes will join in $\hat{R}$ with probability of at least $1 - 1/n^{C_7}$.
That means that all the nodes in $W_1$ have joined either $\hat{R}$ or $\hat{Z}$, with probability of at least $1 - \frac{1}{n^{C_7-1}}$, concluding the proof of Lemma~\ref{lemma:runtime}.

\subsection{Proof of Lemma~\ref{lemma:covered2}}
\label{append:covered2}

Proof by contradiction.
Suppose there exists a node $v \in \hat{Z}$ that is not $\omega_2$-covered by $\hat{R}$.
Therefore, all the nodes in $\hat{R}$ are outside of the ball $B(v, \omega_2)$.
We calculate $SP(\hat{R},v)$ and derive the conflict.
We notice the following three facts:
\begin{inparaenum} [(1)]
\item all the nodes in $\hat{R}$ have a mutual distance of at least $\omega_1$;
\item the distance from any sender node in $\hat{R}$ to $v$ is $>\omega_2$; and
\item $\omega_2 > \omega_1/2 > 0$.
\end{inparaenum}
Due to Proposition~\ref{prop:sensed-interference},
\begin{align*}
SP(\hat{R},v)
\leq & 36 \frac{\alpha-1}{\alpha-2} \frac{P}{\omega_1^2 (\omega_2)^{\alpha-2}} + N \\
\leq & 36 \frac{\alpha-1}{\alpha-2} \frac{\omega_1^{\alpha-2}}{\omega_2^{\alpha-2}} Thres(\omega_1)
\leq  Thres(\omega_1).
\end{align*}

According to the condition for $v$ to enter $\hat{Z}$, we must have $SP(\hat{R},v) > Thres(\omega_1)$, where lies the contradiction.

%\subsection{Appendix to Half Duplex Communication for Ruling Construction}
\subsection{Proof of Lemma~\ref{lemma:spaced}}
\label{append:halfduplex}

Suppose there is ``bad'' node $v \in \hat{R}$ such that there exists a node $v' \in \hat{R}$ and $d(v,v') < \omega_1$.
We call such a $v'$ a ``bad'' partner of $v$.
The only possible situation for $v$ to have a ``bad'' partner $v'$ is that $v$ and $v'$ enter $\hat{R}$ in the same round; otherwise, one of them should have been ``pushed'' into $\hat{Z}$ during the decision step of a round when the other enters $\hat{R}$.

Let $u$ denote an arbitrary node in $\hat{R}$.
The necessary and sufficient condition for $u$ to be bad (or to have a bad partner) is that, in the coordination step of the round when $u$ enters $\hat{R}$, there exists at least one active node $u' \in W_1$ such that,
\begin{inparaenum} [(1)]
\item $d(u,u') < \omega_1$,
\item $u$ and $u'$ are in $U$, and
\item $u$ and $u'$ made the same random binary decisions all through the $C_5 \log m$ slots of the coordination step.
\end{inparaenum}

W.l.o.g., we assume $u$ enters $\hat{R}$ at round $i_{in}$ in phase $i_{out}$.
Let $t_1$ denote the time that round $i_{in}$ in phase $i_{out}$ begins.
The probability that $u$ is bad equals the probability for at least one other active node $u' \in A_{t_1}^{W_1}(u,\omega_1)$ (\ie, $d(v,v') < \omega_1$) to enter $\hat{R}$ at round $i_{in}$ in phase $i_{out}$.
\begin{align*}
& Prob(u \text{ is bad}) \\
\leq & \sum_{u' \in A_{t_1}^{W_1}(u,\omega_1)} Prob(U(u')=1) \frac{1}{2^{C_5 \log n}} \\
\leq & \frac{1}{n^{C_5}} \frac{\bigl|A_{t_1}^{W_1}(u,\omega_1)\bigr|}{2^{\log b_{max} - i_{out} + 2}} \\
\leq & \frac{1}{n^{C_5}} \frac{\max_{\text{active } w \in W_1} |A_{t_1}^{W_1}(w,\omega_1)|}{2^{\log b_{max} - i_{out} + 2}}.
\end{align*}

Further, due to Lemma~\ref{lemma:degree},
\begin{align*}
& Prob\big(\max_{\text{active } w \in W_1} \bigl|A_{t_1}^{W_1}(w,\omega_1)\bigr| \leq 2^{\log b_{max} - i_{out}+1}\big) \\
\geq & 1 - \frac{i_{out}-1}{n^{C_7}}
\geq 1 - \frac{1}{n^{C_7-1}}.
\end{align*}

Therefore,
\begin{align*}
Prob(u \text{ is good}) & \geq (1 - \frac{1}{n^{C_5}})(1 - \frac{1}{n^{C_7-1}}) \\
& \geq 1 - \frac{2}{n^{\min\{C_5, C_7-1\}}}.
\end{align*}

Finally, since $\hat{R}$ contains at most $n$ nodes,
the probability that there are no bad nodes in $\hat{R}$ is at least $1 - \frac{1}{n^{\min\{C_5, C_7-1\}-1}}$.

\end{document}